\documentclass[journal,web]{ieeecolor}

\usepackage{generic}
\usepackage{cite}

\usepackage{amsmath}
\usepackage{amssymb}
\usepackage{amsfonts}
\usepackage{mathtools}
\usepackage{bm}

\renewcommand{\baselinestretch}{0.90}

\renewcommand{\baselinestretch}{0.97}

\usepackage{titlesec}

\titlespacing*{\section}{0pt}{0.8ex}{0.25ex}
\titlespacing*{\subsection}{0pt}{0.45ex}{0.10ex}
\titlespacing*{\subsubsection}{0pt}{0.35ex}{0.08ex}

\usepackage{graphicx}
\usepackage{float}
\usepackage{placeins}
\usepackage{stfloats}

\usepackage[
    caption=false,
    font=normalsize,
    labelfont=sf,
    textfont=sf
]{subfig}

\usepackage{booktabs}
\usepackage{array}
\usepackage{multirow}

\usepackage{algorithm}
\usepackage{algorithmic}

\usepackage{textcomp}
\usepackage{verbatim}
\usepackage{balance}
\usepackage{url}

\usepackage{xcolor}
\usepackage{tikz}
\usepackage{scalerel}

\usetikzlibrary{
    positioning,
    arrows.meta,
    svg.path
}

\makeatletter
\let\IEEEorigmakecaption\@makecaption
\long\def\@makecaption#1#2{%
    {\IEEEorigmakecaption{#1}{#2}}%
}
\makeatother

\def\BibTeX{%
    {\rm B\kern-.05em%
    {\sc i\kern-.025em b}\kern-.08em%
    T\kern-.1667em%
    \lower.7ex\hbox{E}\kern-.125emX}%
}

\definecolor{orcidlogocol}{HTML}{A6CE39}

\tikzset{
    orcidlogo/.pic={
        \fill[orcidlogocol]
        svg{
            M256,128
            c0,70.7-57.3,128-128,128
            C57.3,256,0,198.7,0,128
            C0,57.3,57.3,0,128,0
            C198.7,0,256,57.3,256,128z
        };

        \fill[white]
        svg{
            M86.3,186.2H70.9V79.1h15.4v48.4V186.2z
        }
        svg{
            M108.9,79.1h41.6
            c39.6,0,57,28.3,57,53.6
            c0,27.5-21.5,53.6-56.8,53.6h-41.8V79.1z
            M124.3,172.4h24.5
            c34.9,0,42.9-26.5,42.9-39.7
            c0-21.5-13.7-39.7-43.7-39.7h-23.7V172.4z
        }
        svg{
            M88.7,56.8
            c0,5.5-4.5,10.1-10.1,10.1
            c-5.6,0-10.1-4.6-10.1-10.1
            c0-5.6,4.5-10.1,10.1-10.1
            C84.2,46.7,88.7,51.3,88.7,56.8z
        };
    }
}

\newcommand{\orcidicon}[1]{%
    \href{https://orcid.org/#1}{%
        \mbox{%
            \scalerel*{%
                \begin{tikzpicture}[yscale=-1,transform shape]
                    \pic{orcidlogo};
                \end{tikzpicture}
            }{|}%
        }%
    }%
}

\newtheorem{assumption}{Assumption}
\newtheorem{definition}{Definition}
\newtheorem{lemma}{Lemma}
\newtheorem{theorem}{Theorem}
\newtheorem{remark}{Remark}
\newtheorem{proposition}{Proposition}

\makeatletter

\def\@begintheorem#1#2{%
    \par
    \topsep=1pt plus 0.5pt minus 0.5pt\relax
    \partopsep=0pt\relax
    \itemsep=0pt\relax
    \parsep=0pt\relax
    \trivlist
    \item[\hskip\labelsep
    {\bfseries\color{black}#1\ #2.}]%
    \itshape\color{black}%
}

\def\@opargbegintheorem#1#2#3{%
    \par
    \topsep=1pt plus 0.5pt minus 0.5pt\relax
    \partopsep=0pt\relax
    \itemsep=0pt\relax
    \parsep=0pt\relax
    \trivlist
    \item[\hskip\labelsep
    {\bfseries\color{black}#1\ #2\ (#3).}]%
    \itshape\color{black}%
}

\@ifundefined{proof}{}{%
    \let\IEEEoriginalproof\proof
    \renewcommand{\proof}{%
        \color{black}%
        \IEEEoriginalproof
    }%
}

\makeatother

\newcommand{\spow}[2]{%
    \left\lceil #1\right\rfloor^{#2}%
}

\newcommand{\bx}{\bm x}

\allowdisplaybreaks[3]

\definecolor{myblue}{RGB}{0,80,160}

\makeatletter
\let\NAT@parse\undefined
\makeatother

\usepackage[
    colorlinks=true,
    linkcolor=myblue,
    citecolor=myblue,
    urlcolor=myblue,
    filecolor=myblue,
    pdfborder={0 0 0}
]{hyperref}

\AtBeginDocument{%
    \hypersetup{
        colorlinks=true,
        linkcolor=myblue,
        citecolor=myblue,
        urlcolor=myblue,
        filecolor=myblue,
        pdfborder={0 0 0}
    }%
}

\begin{document}

\title{
Predefined-Time Integral Reinforcement Learning for Saturated Unknown Nonlinear Multi-Agent Systems Under FDI Attacks and Disturbances
}

\author{
Tien Dat Vu and
Minh Doan
\thanks{
T. D. Vu and M. Doan are with the Faculty of Mechanical Engineering,
Ho Chi Minh City University of Technology (HCMUT),
Vietnam National University Ho Chi Minh City (VNU-HCM),
Ho Chi Minh City, Vietnam
(e-mail: dat.vuv@hcmut.edu.vn; minh.doan@hcmut.edu.vn).
}
\thanks{
Corresponding author: Minh Doan
(e-mail: minh.doan@hcmut.edu.vn).
}
}

\maketitle
\maketitle

\begin{abstract}
This paper addresses secure leader--follower formation of unknown nonlinear
multi-agent systems under actuator constraints, external disturbances, and
false-data-injection (FDI) attacks. The graph-coupled coordination-error
dynamics are formulated as local zero-sum differential games, where a
nonquadratic input utility yields saturation-compatible secure policies and
actuator-channel FDI and disturbances act as adversarial inputs. To eliminate
explicit dependence on the unknown nonlinear drift, an integral
Bellman--Isaacs identity enables critic-only learning from finite trajectory
data. A two-power state-cost structure and a deadline-parameterized critic
update connect optimal learning with predefined-time stabilization. Unlike
fixed-time methods whose settling-time bound is determined by preselected
gains, the proposed framework assigns the overall deadline first and allocates
it among data informativity, critic learning, the reinforcement window, and
formation convergence. Practical predefined-time convergence of the critic
and formation errors to bounded residual sets is established independently of
initial conditions, while secure actuator constraints are satisfied by
construction. Simulations validate the framework under FDI attacks,
disturbances, input constraints, and different initial conditions.
\end{abstract}


\section{Introduction}
\label{sec:introduction}

Formation control of networked autonomous systems increasingly operates
under uncertain nonlinear dynamics, actuator limitations, external
disturbances, and adversarial networked conditions. These difficulties are
particularly relevant in cyber--physical multi-agent systems, where exchanged
information and actuator commands may be corrupted by false-data-injection
(FDI) attacks and the resulting effects can propagate through the interaction
graph and disrupt collective motion
\cite{OlfatiSaber2007Consensus,Ren2007InformationConsensus,
MoSinopoli2010FDI,Pasqualetti2013AttackDetection,
Teixeira2015SecureControl}.

Zero-sum differential games provide a natural framework for modeling the
competition between secure control and adversarial channels, with optimality
characterized by the Hamilton--Jacobi--Isaacs (HJI) equation
\cite{Isaacs1965DifferentialGames,BasarOlsder1999DynamicGames}.
Because direct HJI solution is intractable for general unknown nonlinear
systems, adaptive dynamic programming and reinforcement learning approximate
the value function from data
\cite{Vrabie2009DirectAdaptiveOptimalControl,
Vamvoudakis2010OnlineActorCritic}.
Integral reinforcement learning (IRL) further replaces the pointwise
optimality equation by a finite-window Bellman identity, allowing the
unknown drift to be removed from critic regression
\cite{Modares2014OptimalTrackingIRL}. Meanwhile, actuator limits should be
embedded directly into the optimization; nonquadratic input utilities yield
saturation-compatible hyperbolic-tangent policies, while same-channel FDI
can be modeled separately as an adversarial input
\cite{AbuKhalaf2005Nonquadratic}.

Convergence time is equally important. Finite-time stability permits
initial-condition-dependent settling times
\cite{BhatBernstein2000FiniteTime}, whereas fixed-time stability guarantees
a uniform bound independent of the initial state
\cite{Polyakov2012FixedTime}. That bound, however, is generally induced by
chosen gains rather than assigned beforehand. Predefined-time stability
reverses this design logic by prescribing the deadline first and selecting
the gains accordingly. Kokolakis \emph{et al.} developed predefined-time RL
for nonlinear optimal feedback control
\cite{Kokotakis2025PredefinedTimeRL}, but only for a single nonlinear system,
without graph-coupled formation, actuator-channel FDI, or saturation-aware
multi-agent interaction. Gong \emph{et al.} established fixed-time RL for
secure formation of second-order multi-agent systems under FDI attacks
\cite{Gong2025SecureFormationFDI}, but the convergence horizon is fixed-time
rather than designer-assigned predefined-time, and general unknown nonlinear
saturated agents are not addressed.

Our recent works resolved complementary parts of this problem. Fixed-time
resilient IRL was first developed for unknown input-constrained nonlinear
systems under FDI attacks and disturbances
\cite{VuDoan2026FixedTimeResilientIRL}, and was subsequently strengthened to
predefined-time resilient IRL with designer-assigned learning and
closed-loop horizons
\cite{VuDoan2026PredefinedTimeResilientIRL}; both concern single nonlinear
systems. In parallel, fixed-time IRL was extended to saturated nonlinear
leader--follower multi-agent systems under FDI attacks
\cite{VuDoan2026FixedTimeMASIRL}. That formulation achieves secure formation
within an initial-condition-independent fixed-time bound, but the bound is
determined by the selected gains and cannot be assigned directly a priori.
Thus, a unified multi-agent framework combining unknown nonlinear dynamics,
bounded actuation, adversarial channels, data-driven HJI learning, and a
designer-assigned convergence deadline remains missing.

This paper develops a predefined-time resilient IRL framework to close this
gap. The graph-coupled coordination-error dynamics are decomposed into
unknown drift, self and neighboring secure-control channels, and local and
neighboring adversarial channels. A saturation-aware local zero-sum HJI game
is combined with a two-power state cost that generates the dissipation
structure required for predefined-time analysis. An integral Bellman--Isaacs
residual removes the unknown drift from critic learning, while finite replay
data preserve informativity after online excitation fades. The total
convergence horizon is assigned first and then explicitly allocated among
data informativity, critic learning, the reinforcement window, and formation
convergence.


The main contributions are fourfold.
\emph{First}, an end-to-end resilient architecture integrates graph-coupled secure formation, saturation-aware graphical games, IRL, FDI/disturbance attenuation, and predefined-time certification for unknown nonlinear multi-agent systems.
\emph{Second}, a cost-induced predefined-time HJI construction employs low- and high-order state penalties to generate the required two-power dissipation without imposing restrictive conditions on the unknown ideal value function.
\emph{Third}, a critic-only integral learning law enables explicit gain synthesis from an assigned critic-learning horizon, with finite replay informativity replacing persistent excitation.
\emph{Fourth}, a unified pointwise--integral Lyapunov analysis combines a derivative-limited moving-window inequality with explicit deadline allocation to establish practical predefined-time convergence of both learning and formation while preserving actuator constraints by construction.


\textbf{Notation:}
The notation used in this paper is standard. For a vector \(\mathbf z\), \(\|\mathbf z\|\) denotes the Euclidean norm. For a symmetric matrix \(A\), \(A>0\) (\(A\ge0\)) denotes positive definiteness (positive semidefiniteness). For a matrix \(\mathbf M\), \(\|\mathbf M\|\) denotes the induced norm; if \(\mathbf M=\mathbf M^\top\), then \(\lambda_{\min}(\mathbf M)\) and \(\lambda_{\max}(\mathbf M)\) denote its minimum and maximum eigenvalues. The symbols \(\operatorname{col}(\cdot)\), \(\operatorname{diag}(\cdot)\), and \(\mathbf I_n\) denote column stacking, a diagonal matrix, and the \(n\times n\) identity matrix, respectively, while \(\mathbf 1_n\in\mathbb R^n\) denotes the \(n\)-dimensional all-ones vector. The notation \(AC_{\mathrm{loc}}([0,\infty);\mathbb R^n)\) denotes the space of locally absolutely continuous functions from \([0,\infty)\) to \(\mathbb R^n\), i.e., functions absolutely continuous on every compact subinterval. A continuous function \(\alpha:[0,\infty)\to[0,\infty)\) belongs to class \(\mathcal K_\infty\) if \(\alpha(0)=0\), \(\alpha\) is strictly increasing, and \(\alpha(s)\to\infty\) as \(s\to\infty\). For any \(r>0\) and \(x_c\in\mathbb R^n\), \(\mathcal B_r(x_c):=\{x\in\mathbb R^n:\|x-x_c\|<r\}\) and \(\overline{\mathcal B}_r(x_c):=\{x\in\mathbb R^n:\|x-x_c\|\le r\}\) denote the open and closed Euclidean balls of radius \(r\) centered at \(x_c\), respectively.

\section{Preliminaries and Problem Formulation}

\subsection{Preliminaries}

Consider a graph \(\mathcal G=(\mathcal V,\mathcal E,\mathbf A)\), where \(\mathcal V=\{1,\ldots,N\}\), \(\mathcal E\subseteq\mathcal V\times\mathcal V\), and \(\mathbf A=[a_{ij}]\in\mathbb R^{N\times N}\) is the adjacency matrix. The weight \(a_{ij}>0\) means that follower \(i\) receives information from follower \(j\), and \(a_{ij}=0\) otherwise. Define \(\mathcal N_i:=\{j\in\mathcal V:(j,i)\in\mathcal E\}\), \(\mathbf L:=\mathbf D_{\mathcal G}-\mathbf A\), \(\mathbf D_{\mathcal G}:=\operatorname{diag}\{\sum_{j=1}^{N}a_{ij}\}\), \(\mathbf B_0:=\operatorname{diag}(b_{10},\ldots,b_{N0})\), and \(\mathbf H:=\mathbf L+\mathbf B_0\), where \(b_{i0}\ge0\) is the pinning gain from the leader to follower \(i\). This leader-rooted construction follows the secure formation-control setting of \cite{Gong2025SecureFormationFDI} and will be used to define the distributed formation error.

\begin{assumption}[Graph connectivity \cite{Hong2008DistributedObservers,BermanPlemmons1994}]
\label{ass:graph_connectivity}
The leader is globally reachable from the follower graph. Equivalently, \(\mathbf H\) is nonsingular. For a directed graph, there exists a positive diagonal matrix \(\mathbf\Pi:=\operatorname{diag}(\pi_1,\ldots,\pi_N)\) such that \(\mathbf H_s:=\frac{\mathbf\Pi\mathbf H+\mathbf H^\top\mathbf\Pi}{2}>0\).
\end{assumption}

\begin{definition}[Symmetric input constraint]
\label{def:symmetric_input_constraint}
For each follower \(i\), the actuator input is said to satisfy a
symmetric input constraint if
\begin{equation}
    u_{ci}\in\mathbb U_i
    :=
    \left\{
        u_{ci}\in\mathbb R^m:
        |u_{ci\ell}|<\bar u_{i\ell},\
        \ell=1,\ldots,m
    \right\},
    \label{eq:symmetric_input_constraint}
\end{equation}
where \(\bar u_{i\ell}>0\) is the admissible magnitude of the
\(\ell\)-th input channel.
\begin{equation}
    \bar{\mathbf U}_i
    =
    \operatorname{diag}
    (\bar u_{i1},\ldots,\bar u_{im}),
    \label{eq:Ubar_def}
\end{equation}
\end{definition}

This symmetric bounded-input setting is standard in constrained optimal
control and saturation-aware ADP/HJI designs
\cite{AbuKhalaf2006HJIInputSaturation,
AbuKhalaf2005Nonquadratic,
Bai2019AdaptiveRLSaturation}. It will be used in Section~\ref{Sec3} to
construct a nonquadratic saturation-compatible utility and to obtain a
bounded HJI policy satisfying \(u_{ci}\in\mathbb U_i\) by construction.

\subsection{Problem Formulation}

Consider a leader--follower network with one leader and \(N\) nonlinear
followers. The leader and follower \(i\) satisfy, respectively,
\begin{equation}
\dot x_0=f_0(x_0),\qquad x_0\in\mathbb R^n,
\label{eq:leader_model}
\end{equation}
and
\begin{equation}
\dot x_i
=
f_i(x_i)
+
\mathbf g_i(x_i)u_i
+
\mathbf d_i(x_i)\omega_i,
\qquad
i=1,\ldots,N,
\label{eq:follower_model}
\end{equation}
where \(x_i\in\mathbb R^n\), \(u_i\in\mathbb R^m\), and
\(\omega_i\in\mathbb R^{n_{\omega_i}}\) is an unknown disturbance.
Following the control-channel FDI model in
\cite{Gong2025SecureFormationFDI}, let
\begin{equation}
u_i(t)=u_{ci}(t)+u_{ai}(t),
\label{eq:input_decomposition}
\end{equation}
where \(u_{ci}\in\mathbb R^m\) is the defender-generated secure control
and \(u_{ai}\in\mathbb R^m\) is the FDI signal. Assume
\(\omega_i,u_{ai}\in L_2([0,\infty))\cap L_\infty([0,\infty))\),
with \(\|\omega_i(t)\|\le\bar\omega_i\) and
\(\|u_{ai}(t)\|\le\bar u_{ai}\) for all \(t\ge0\). Thus,
\begin{equation}
\dot x_i
=
f_i(x_i)
+
\mathbf g_i(x_i)u_{ci}
+
\mathbf g_i(x_i)u_{ai}
+
\mathbf d_i(x_i)\omega_i .
\label{eq:attacked_follower_model}
\end{equation}
The symmetric input constraint of
Definition~\ref{def:symmetric_input_constraint} is imposed on the
admissible secure policy in the HJI design; no saturation map is included
in the plant model.

\begin{remark}[FDI location and actuator constraint]
\label{rem:fdi_location}
The model in \eqref{eq:attacked_follower_model} assumes that the actuator-side FDI
signal \(u_{ai}\) is injected after the bounded secure command \(u_{ci}\)
has been generated. Hence, the constraint
\(|u_{ci,\ell}|<\bar u_{i\ell}\) applies to the defender-generated command
only, whereas the plant receives \(u_{ci}+u_{ai}\). If the attack were
instead injected before a physical saturation element, the plant input
would be \(\operatorname{sat}(u_{ci}+u_{ai})\), which would lead to a
different game formulation. The present work adopts the former architecture.
\end{remark}

\begin{definition}[$\mathcal{L}_2$-gain]
\label{def2}
Consider a nonlinear system with exogenous input \(v\in\mathbb{R}^{n_v}\) and regulated output \(y_r\in\mathbb{R}^{n_r}\). The system has finite \(\mathcal{L}_2\)-gain not exceeding \(\lambda>0\) if there exists a nonnegative function \(\alpha(\cdot)\), with \(\alpha(0)=0\), such that every admissible trajectory satisfies \(\int_0^T\|y_r(t)\|^2dt\le\lambda^2\int_0^T\|v(t)\|^2dt+\alpha(x(0))\) for all \(T\ge0\). For \(x(0)=0\), this gives \(\|y_r\|_2\le\lambda\|v\|_2\). Equivalently, the bound is certified by a storage function \(\mathcal{S}(x)\ge0\) satisfying \(\dot{\mathcal{S}}+\|y_r\|^2-\lambda^2\|v\|^2\le0\), which provides the standard link to the corresponding HJB/HJI formulation.
\end{definition}

\begin{assumption}[Regularity]
\label{ass2}
Let \(\Omega_i\subset\mathbb R^n\) be a compact admissible operating region containing the origin. The functions \(f_0(x_0)\), \(f_i(x_i)\), \(\mathbf g_i(x_i)\), and \(\mathbf d_i(x_i)\) are locally Lipschitz on \(\Omega_i\). The desired offsets \(h_i(t)\) are continuously differentiable, and \(h_i(t)\), \(\dot h_i(t)\) are bounded.
\label{ass:regularity}
\end{assumption}

\begin{definition}[Admissible secure policy \cite{AbuKhalafLewis2005,Vamvoudakis2010OnlineActorCritic}]
For the local graphical game defined above, a secure control policy \(u_{ci}=u_{ci}(x_i)\) is admissible on \(\Omega_i\), denoted by \(u_{ci}\in\Psi_i(\Omega_i)\), if \(u_{ci}(x_i)\) is continuous on \(\Omega_i\), \(u_{ci}(0)=0\), \(u_{ci}(x_i)\in\mathbb U_i\) for all \(x_i\in\Omega_i\), the disturbance-and-attack local coordination-error dynamics under \(u_{ci}\) are asymptotically stable on \(\Omega_i\), and the associated infinite-horizon cost \(J_i\) is finite for every \(x_i(0)\in\Omega_i\).
\end{definition}
\begin{remark}
Instead of merely assuming an initial admissible policy as in classical
PI/ADP formulations
\cite{AbuKhalafLewis2005,Vrabie2009DirectAdaptiveOptimalControl},
Appendix~A in \cite{VuDoan2026FixedTimeMASIRL} reuses the replay data collected by
the IRL mechanism to construct and certify such a policy.
\end{remark}

Let \(h_i(t)\in\mathbb R^n\) be the desired offset of follower \(i\), and
define
\(
e_i:=x_i-x_0-h_i
\)
and
\(
\chi_i:=
\sum_{j\in\mathcal N_i}a_{ij}(e_i-e_j)+b_{i0}e_i
\).
Equivalently,
\(
\chi_i=
\sum_{j\in\mathcal N_i}
a_{ij}[(x_i-h_i)-(x_j-h_j)]
+b_{i0}[(x_i-h_i)-x_0]
\).
With
\(e:=\operatorname{col}(e_1,\ldots,e_N)\) and
\(\chi:=\operatorname{col}(\chi_1,\ldots,\chi_N)\),
\(
\chi=(\mathbf H\otimes\mathbf I_n)e
\).
Under Assumption~\ref{ass:graph_connectivity},
\(
\chi=0\Longleftrightarrow e=0
\)
and
\(
\|e\|
\le
\|(\mathbf H^{-1}\otimes\mathbf I_n)\|\,\|\chi\|
\).
Hence, formation tracking is studied through the stabilization of
\(\chi\).

For compactness, let
\(
F_i^e:=
f_i(x_i)-f_0(x_0)-\dot h_i
\),
\(
\mathbf g_i:=\mathbf g_i(x_i)
\), and
\(
\mathbf d_i:=\mathbf d_i(x_i)
\).
From \eqref{eq:attacked_follower_model},
\(
\dot e_i=
F_i^e+\mathbf g_i u_{ci}
+\mathbf g_i u_{ai}
+\mathbf d_i\omega_i
\),
and
\(
\dot\chi_i=
\sum_{j\in\mathcal N_i}a_{ij}(\dot e_i-\dot e_j)
+b_{i0}\dot e_i
\).
Therefore,
\(
\dot{\chi}_i
=
\sum_{j\in\mathcal N_i}
a_{ij}
(
F_i^e-F_j^e
+\mathbf g_i u_{ci}
-\mathbf g_j u_{cj}
)
+
b_{i0}
(
F_i^e+\mathbf g_i u_{ci}
)
+
\sum_{j\in\mathcal N_i}
a_{ij}
(
\mathbf g_i u_{ai}
-\mathbf g_j u_{aj}
+\mathbf d_i\omega_i
-\mathbf d_j\omega_j
)
+
b_{i0}
(
\mathbf g_i u_{ai}
+\mathbf d_i\omega_i
)
\).

Define
\begin{equation}
\kappa_i
=
b_{i0}
+
\sum_{j\in\mathcal N_i}a_{ij},
\qquad
\mathbf G_{ii}^{\chi}(x_i)
=
\kappa_i\mathbf g_i.
\label{eq:kappa_self_G}
\end{equation}
Let
\(\mathcal N_i=\{j_1,\ldots,j_{n_i^{\mathrm{nb}}}\}\),
\(n_i^{\mathrm{nb}}:=|\mathcal N_i|\),
\(x_{\mathcal N_i}:=
\operatorname{col}(x_{j_1},\ldots,x_{j_{n_i^{\mathrm{nb}}}})\), and
\(u_{c\mathcal N_i}:=
\operatorname{col}(u_{cj_1},\ldots,u_{cj_{n_i^{\mathrm{nb}}}})\).
Then
\begin{equation}
\mathbf G_{i\mathcal N}^{\chi}(x_{\mathcal N_i})
=
\big[
-a_{ij_1}\mathbf g_{j_1},
\ldots,
-a_{ij_{n_i^{\mathrm{nb}}}}
\mathbf g_{j_{n_i^{\mathrm{nb}}}}
\big],
\label{eq:neighbor_G}
\end{equation}
where
\(\mathbf g_{j_\ell}=\mathbf g_{j_\ell}(x_{j_\ell})\), so that
\(
\mathbf G_{i\mathcal N}^{\chi}(x_{\mathcal N_i})
u_{c\mathcal N_i}
=
-\sum_{j\in\mathcal N_i}
a_{ij}\mathbf g_j(x_j)u_{cj}
\).

Define
\begin{equation}
F_i^{\chi}
=
\sum_{j\in\mathcal N_i}
a_{ij}(F_i^e-F_j^e)
+
b_{i0}F_i^e.
\label{eq:FX_def}
\end{equation}
Equivalently,
\(
F_i^{\chi}
=
\kappa_i
[f_i(x_i)-f_0(x_0)-\dot h_i]
-
\sum_{j\in\mathcal N_i}
a_{ij}
[f_j(x_j)-f_0(x_0)-\dot h_j]
\),
which collects the nonlinear drift, leader dynamics, and desired-offset
dynamics.

Likewise,
\(
\Xi_i^{\chi}
=
\kappa_i
(\mathbf g_i u_{ai}+\mathbf d_i\omega_i)
-
\sum_{j\in\mathcal N_i}
a_{ij}
(\mathbf g_j u_{aj}+\mathbf d_j\omega_j)
\).
Define
\begin{equation}
\nu_i
=
\operatorname{col}
\big(
\omega_i,
\omega_{j_1},
\ldots,
\omega_{j_{n_i^{\mathrm{nb}}}},
u_{ai},
u_{aj_1},
\ldots,
u_{aj_{n_i^{\mathrm{nb}}}}
\big).
\label{eq:nu_def}
\end{equation}
A compatible adversarial input matrix is
\begin{align}
\mathbf D_i^{\chi}(x_i,x_{\mathcal N_i})
=
\big[
&\kappa_i\mathbf d_i,\,
-a_{ij_1}\mathbf d_{j_1},
\ldots,
-a_{ij_{n_i^{\mathrm{nb}}}}
\mathbf d_{j_{n_i^{\mathrm{nb}}}},
\nonumber\\
&\kappa_i\mathbf g_i,\,
-a_{ij_1}\mathbf g_{j_1},
\ldots,
-a_{ij_{n_i^{\mathrm{nb}}}}
\mathbf g_{j_{n_i^{\mathrm{nb}}}}
\big].
\label{eq:DX_matrix}
\end{align}
Thus,
\begin{equation}
\Xi_i^{\chi}
=
\mathbf D_i^{\chi}(x_i,x_{\mathcal N_i})\nu_i.
\label{eq:Xi_Dnu}
\end{equation}

Therefore,
\begin{align}
\dot\chi_i
&=
F_i^{\chi}(\xi_i)
+
\mathbf G_{ii}^{\chi}(x_i)u_{ci}
+
\mathbf G_{i\mathcal N}^{\chi}(x_{\mathcal N_i})
u_{c\mathcal N_i}
\nonumber\\
&\quad
+
\mathbf D_i^{\chi}(x_i,x_{\mathcal N_i})\nu_i,
\label{eq:X_dynamics_final}
\end{align}
where
\begin{equation}
\xi_i
=
\operatorname{col}
\big(
x_i,
x_{\mathcal N_i},
x_0,
\dot h_i,
\dot h_{\mathcal N_i}
\big),
\qquad
\dot h_{\mathcal N_i}
=
\operatorname{col}
\big(
\dot h_{j_1},
\ldots,
\dot h_{j_{n_i^{\mathrm{nb}}}}
\big).
\label{eq:zeta_def}
\end{equation}

\begin{remark}[Reconstruction of the admissible operating region]
\label{remark3}
The origin is retained as the nominal regulation target and is not excluded from the operating region \(\Omega_i\). For the fixed-time comparison analysis, choose an arbitrary \(r_{i,-}>0\) such that \(B_{r_{i,-}}(0)\subset\Omega_i\), consistently with the practical terminal region established below. Three mutually exclusive cases are considered. If \(\chi_i=0\), the nominal coordination objective has already been achieved and, by admissibility, \(u_{ci}(0)=0\) while the associated infinite-horizon value is finite, so no convergence estimate is required. If \(0<\|\chi_i\|<r_{i,-}\), the error already lies in the prescribed terminal neighborhood \(B_{r_{i,-}}(0)\), and the fixed-time comparison inequalities used to prove entrance into that neighborhood need not be invoked. The nontrivial case is therefore \(\|\chi_i\|\ge r_{i,-}\), for which define \(\Omega_i^{r}:=\{\chi_i\in\Omega_i:\|\chi_i\|\ge r_{i,-}\}\). Since \(\Omega_i\) is compact and \(\{\chi_i:\|\chi_i\|\ge r_{i,-}\}\) is closed, \(\Omega_i^{r}\) is compact, \(0\notin\Omega_i^{r}\), and \(\inf_{\chi_i\in\Omega_i^{r}}\|\chi_i\|\ge r_{i,-}>0\). Consequently, \(V_i^{*}(\chi_i)/\|\chi_i\|^{2}\) and \(\|\nabla_{\chi_i}V_i^{*}(\chi_i)\|/\|\chi_i\|\) are well defined on \(\Omega_i^{r}\) (see the proof of Lemma~\ref{lem:local_value_bounds}), which is precisely the region where the fixed-time decay argument is required. If a disturbance or FDI signal drives the trajectory away from the origin, no exclusion assumption is imposed: while \(0<\|\chi_i\|<r_{i,-}\), the trajectory remains inside the prescribed terminal neighborhood, and once \(\|\chi_i\|\ge r_{i,-}\), the same comparison argument applies again. Thus, excluding the origin from \(\Omega_i^{r}\) is purely analytical and does not remove it from either the HJI operating domain or the admissible-policy formulation.
\end{remark}

\section{Secure Saturation-Aware HJI Formulation}
\label{Sec3}

\providecommand{\bx}{\bm x}


\subsection{Graphical Game and Predefined-Time Shaped HJI}

\begingroup
\setlength{\abovedisplayskip}{3pt}
\setlength{\belowdisplayskip}{3pt}
\setlength{\abovedisplayshortskip}{2pt}
\setlength{\belowdisplayshortskip}{2pt}
\setlength{\jot}{1pt}

Accordingly, the lumped term in \eqref{eq:X_dynamics_final} can be equivalently decomposed as
\begin{equation}
\mathbf D_i^\chi(x_i,x_{\mathcal N_i})\nu_i
=
\mathbf B_i^\chi(x_i)\varpi_i
+
\sum_{j\in\mathcal N_i}
\mathbf B_{ij}^\chi(x_j)\varpi_j.
\label{eq:adversarial_channel_decomposition}
\end{equation}
Here, \(\mathbf B_i^\chi(x_i)=\kappa_i[\mathbf d_i(x_i)\ \mathbf g_i(x_i)]\), \(\varpi_i=\operatorname{col}\{\omega_i,u_{ai}\}\), \(\varpi_j=\operatorname{col}\{\omega_j,u_{aj}\}\), and \(\mathbf B_{ij}^\chi(x_j)=[-a_{ij}\mathbf d_j(x_j)\ -a_{ij}\mathbf g_j(x_j)]\). This decomposition distinguishes the locally acting adversarial signal from those entering through neighboring agents and directly yields the local graphical-game representation
\begin{align}
\dot{\chi}_i
={}&
F_i^{\chi}(\xi_i)
+\mathbf G_{ii}^{\chi}(x_i)u_{ci}
+\mathbf G_{i\mathcal N}^{\chi}(x_{\mathcal N_i})u_{c\mathcal N_i}
\nonumber\\[-1mm]
&+
\mathbf B_i^{\chi}(x_i)\varpi_i
+\sum_{j\in\mathcal N_i}
\mathbf B_{ij}^{\chi}(x_j)\varpi_j.
\label{eq:hji_local_dynamics_game}
\end{align}

\endgroup

\begin{assumption}[Local state sufficiency/Markov condition]
\label{assk}
For each follower \(i\), on the operating region \(\Omega_i\), any two admissible physical configurations associated with the same \(\chi_i\), under identical control and adversarial inputs, are assumed to induce the same right-hand side of \eqref{eq:hji_local_dynamics_game}.
\label{ass:local_state_sufficiency}
\end{assumption}

Thus, the last two terms in \eqref{eq:hji_local_dynamics_game} are exactly the decomposed form of the lumped adversarial channel \(\mathbf D_i^{\chi}(x_i,x_{\mathcal N_i})\nu_i\) in \eqref{eq:X_dynamics_final}. The neighboring secure inputs \(u_{c\mathcal N_i}\) are treated as fixed local coupling signals when the HJI equation of follower \(i\) is formed. For a fixed neighboring secure policy \(u_{c,-i}:=\{u_{cj}\}_{j\in\mathcal N_i}\), define the local graphical-game cost (From~Definition~\ref{def2})
\begin{align}
J_i
&=
\int_0^\infty
\Big[
Q_{ii}(\chi_i)
+\mathcal U_i(u_{ci})
+\sum_{j\in\mathcal N_i}\mathcal U_{ij}(u_{cj})
\nonumber\\
&\qquad
-\gamma_i^2\varpi_i^\top\mathbf T_{ii}\varpi_i
-\gamma_i^2
\sum_{j\in\mathcal N_i}
\varpi_j^\top\mathbf T_{ij}\varpi_j
\Big]dt ,
\label{eq:graphical_game_cost}
\end{align}
where \(Q_{ii}:\mathbb R^n\to\mathbb R_{\ge0}\) is a continuous positive-definite state penalty satisfying
\begin{equation}
Q_{ii}(0)=0,
\qquad
Q_{ii}(\chi_i)>0,
\quad
\forall\,\chi_i\neq0,
\label{eq:Qii_general_condition}
\end{equation}
\(\mathbf T_{ii}=\mathbf T_{ii}^{\top}>0\), \(\mathbf T_{ij}=\mathbf T_{ij}^{\top}>0\), and \(\gamma_i>0\). The functions \(\mathcal U_i\) and \(\mathcal U_{ij}\) penalize the secure control efforts of follower \(i\) and its neighbors, whereas the negative quadratic terms represent the maximizing role of the disturbances and same-channel FDI attacks.

\begin{remark}
The state penalty \(Q_{ii}\) is deliberately kept general at this stage. Its predefined-time shaping is introduced only after the nominal HJI equation has been obtained. Hence, the graphical-game formulation itself does not assume a prescribed decay condition on the unknown value function and does not require any value-dependent quantity to be inserted artificially into the running cost.
\end{remark}


Define \(\varpi_{-i}:=\{\varpi_j:j\in\mathcal N_i\}\). The corresponding local value function is
\begin{equation}
V_i^\ast(\chi_i(0))
=
\min_{u_{ci}}
\max_{\varpi_i,\varpi_{-i}}
J_i\big(
\chi_i(0),
u_{ci},
u_{c,-i},
\varpi_i,
\varpi_{-i}
\big).
\label{eq:nominal_game_value}
\end{equation}


\begin{lemma}[Nash saddle condition]
\label{lem:nash_saddle_condition}
A set of policies \(\{u_{ci}^\ast,\varpi_i^\ast\}_{i=1}^{N}\) is a Nash saddle equilibrium if, for every follower \(i\), \(J_i(u_{ci}^\ast,u_{c,-i}^\ast,\varpi_i,\varpi_{-i})\le J_i(u_{ci}^\ast,u_{c,-i}^\ast,\varpi_i^\ast,\varpi_{-i}^\ast)\le J_i(u_{ci},u_{c,-i}^\ast,\varpi_i^\ast,\varpi_{-i}^\ast)\) for all admissible \(u_{ci}\), \(\varpi_i\), and \(\varpi_{-i}\).
\end{lemma}

\begin{remark}
Lemma~\ref{lem:nash_saddle_condition} is the standard saddle-point characterization of a zero-sum dynamic game; see \cite{BasarOlsder1999DynamicGames}. In the present context, the defender cannot decrease the local cost by changing \(u_{ci}\) alone, whereas the adversarial player cannot increase it by changing \(\varpi_i\) alone once the saddle policies are reached.
\end{remark}

The Hamiltonian associated with the original cost \eqref{eq:graphical_game_cost} is \(\bar{\mathcal H}_i=Q_{ii}(\chi_i)+\mathcal U_i(u_{ci})+\sum_{j\in\mathcal N_i}\mathcal U_{ij}(u_{cj})-\gamma_i^2\varpi_i^\top\mathbf T_{ii}\varpi_i-\gamma_i^2\sum_{j\in\mathcal N_i}\varpi_j^\top\mathbf T_{ij}\varpi_j+\nabla V_i^{\ast\top}[F_i^\chi(\xi_i)+\mathbf G_{ii}^\chi(x_i)u_{ci}+\mathbf G_{i\mathcal N}^\chi(x_{\mathcal N_i})u_{c\mathcal N_i}+\mathbf B_i^\chi(x_i)\varpi_i+\sum_{j\in\mathcal N_i}\mathbf B_{ij}^\chi(x_j)\varpi_j]\).

\begin{proposition}[Nominal local HJI equation]
\label{prop:nominal_hji}
If \(V_i^\ast\) is continuously differentiable and is the value function of \eqref{eq:nominal_game_value}, then \(V_i^\ast\) satisfies
\begin{equation}
0
=
\min_{u_{ci}\in\mathbb U_i}
\max_{\varpi_i,\varpi_{-i}}
\bar{\mathcal H}_i
(
\chi_i,
\nabla V_i^\ast,
u_{ci},
\varpi_i,
\varpi_{-i}
).
\label{eq:nominal_hji_equation}
\end{equation}
\end{proposition}

\begin{theorem}[Normalized predefined-time comparison~\cite{VuDoan2026PredefinedTimeResilientIRL}]
\label{thm:normalized_predefined_time}
Let \(V:\mathbb R^n\to\mathbb R_{\geq0}\) be positive definite, radially unbounded, and locally absolutely continuous along the trajectories of a forward-complete system. Let \(\alpha,\beta,p,q,r,T_s>0\) satisfy \(0<p<q\) and \(pr<1<qr\), and define \(\gamma_{p,q,r}:=\dfrac{\Gamma_{\!E}\!\left(\frac{1-pr}{q-p}\right)\Gamma_{\!E}\!\left(\frac{qr-1}{q-p}\right)}{\alpha^r\Gamma_{\!E}(r)(q-p)}\left(\dfrac{\alpha}{\beta}\right)^{\frac{1-pr}{q-p}}\), where \(\Gamma_{\!E}(z):=\int_0^\infty t^{z-1}e^{-t}dt\), \(z>0\), is the Euler gamma function. If \(\dot V(\bx(t))\leq-\dfrac{\gamma_{p,q,r}}{T_s}\left[\alpha V^p(\bx(t))+\beta V^q(\bx(t))\right]^r\) for almost all \(t\) such that \(V(\bx(t))>0\), then the origin is globally predefined-time stable and \(T(\bx_0)\leq T_s\) for all \(\bx_0\in\mathbb R^n\).
\end{theorem}

\begin{lemma}[Exact two-power settling-time integral]
\label{lem:Gamma_comparison}
Let \(Y:[t_0,\infty)\rightarrow\mathbb R_{\geq0}\) be absolutely continuous and satisfy \(\dot Y\le-aY^{\frac{\gamma_1+1}{2}}-bY^{\frac{\gamma_2+1}{2}}\) almost everywhere, where \(a,b>0\) and \(0<\gamma_1<1<\gamma_2\). Then \(Y\) reaches zero in a time satisfying \(T(Y(t_0))\le\displaystyle\int_{0}^{Y(t_0)}\frac{ds}{a s^{\frac{\gamma_1+1}{2}}+b s^{\frac{\gamma_2+1}{2}}}<\dfrac{2\Gamma_{\!E}\!\left(\frac{\gamma_2-1}{\gamma_2-\gamma_1}\right)\Gamma_{\!E}\!\left(\frac{1-\gamma_1}{\gamma_2-\gamma_1}\right)}{(\gamma_2-\gamma_1)a^{\frac{\gamma_2-1}{\gamma_2-\gamma_1}}b^{\frac{1-\gamma_1}{\gamma_2-\gamma_1}}}\), where the last bound is independent of \(Y(t_0)\). Moreover, this bound is the \(r=1\) two-power specialization of the normalized predefined-time comparison in Theorem~\ref{thm:normalized_predefined_time}.
\end{lemma}

\begin{proof}
Whenever \(Y>0\), the dissipation inequality gives \(dt\le-dY/(aY^{\frac{\gamma_1+1}{2}}+bY^{\frac{\gamma_2+1}{2}})\), and integration from \(Y(t_0)\) to \(0\) yields \(T(Y(t_0))\le\displaystyle\int_{0}^{Y(t_0)}\frac{ds}{a s^{\frac{\gamma_1+1}{2}}+b s^{\frac{\gamma_2+1}{2}}}\). For the uniform bound, apply Theorem~\ref{thm:normalized_predefined_time} with \(r=1\), \(p=(\gamma_1+1)/2\), \(q=(\gamma_2+1)/2\), \(\alpha=a\), and \(\beta=b\). Since \(q-p=(\gamma_2-\gamma_1)/2\), \((1-p)/(q-p)=(1-\gamma_1)/(\gamma_2-\gamma_1)\), and \((q-1)/(q-p)=(\gamma_2-1)/(\gamma_2-\gamma_1)\), the normalized Gamma factor becomes exactly \(\dfrac{2\Gamma_{\!E}\!\left(\frac{\gamma_2-1}{\gamma_2-\gamma_1}\right)\Gamma_{\!E}\!\left(\frac{1-\gamma_1}{\gamma_2-\gamma_1}\right)}{(\gamma_2-\gamma_1)a^{\frac{\gamma_2-1}{\gamma_2-\gamma_1}}b^{\frac{1-\gamma_1}{\gamma_2-\gamma_1}}}\). Hence the settling-time bound is finite and independent of \(Y(t_0)\).
\end{proof}

\begin{remark}
\label{rem:two_power_quantitative}
Lemma~\ref{lem:Gamma_comparison} addresses the converse design question associated with Theorem~\ref{thm:normalized_predefined_time}: instead of asking whether a prescribed differential inequality guarantees a selected settling time, it evaluates the maximum uniform convergence time generated by the coefficients of a two-power differential equation and therefore permits the coefficients to be recovered from a desired time budget. Let \(A_\gamma:=(\gamma_2-1)/(\gamma_2-\gamma_1)\) and \(B_\gamma:=(1-\gamma_1)/(\gamma_2-\gamma_1)\), so that \(A_\gamma,B_\gamma\in(0,1)\) and \(A_\gamma+B_\gamma=1\). Lemma~\ref{lem:Gamma_comparison} can then be written as \(T(Y(t_0))<C_\gamma a^{-A_\gamma}b^{-B_\gamma}\), where \(C_\gamma:=2\Gamma_{\!E}(A_\gamma)\Gamma_{\!E}(B_\gamma)/(\gamma_2-\gamma_1)\). Hence \(a\) and \(b\) quantify, respectively, the low- and high-power contributions to the uniform settling-time bound, while their balance occurs at \(Y_\times=(a/b)^{2/(\gamma_2-\gamma_1)}\). In particular, the common scaling \((a,b)\mapsto(ca,cb)\) gives \(T\mapsto T/c\). For the critic dynamics developed later, with \(\mu=(\gamma_1+1)/2\), \(\nu=(\gamma_2+1)/2\), \(a=a_0(2\alpha_c)^\mu\), and \(b=b_0(2\alpha_c)^\nu\), one has \(\mu A_\gamma+\nu B_\gamma=1\); consequently, the critic convergence-time bound scales as \(1/\alpha_c\). This identity is the quantitative basis for selecting the critic learning gain directly from the assigned horizon rather than selecting the gain first and computing the settling-time bound afterward.
\end{remark}


\begin{lemma}[Local bounds of the ideal HJI value function]
\label{lem:local_value_bounds}
Suppose that \(V_i^*\) is positive definite on \(\Omega_i\), satisfies \(V_i^*(0)=0\), and \(V_i^*\in C^1(\Omega_i\setminus\{0\})\). Then there exist a function \(\underline{\alpha}_i\in\mathcal K_{\infty}\) and constants \(\bar c_i>0\) and \(c_{\nabla i}>0\) such that \(\underline{\alpha}_i(\|\chi_i\|)\le V_i^*(\chi_i)\le\bar c_i\|\chi_i\|^2\) and \(\|\nabla_{\chi_i}V_i^*(\chi_i)\|\le c_{\nabla i}\|\chi_i\|\) for all \(\chi_i\in\Omega_i^{r}\).
\end{lemma}

\begin{proof}
From Remark~\ref{remark3}, Since \(\Omega_i\) is compact and \(\{\chi_i:\|\chi_i\|\ge r_{i,-}\}\) is closed, \(\Omega_i^{r}\) is compact and \(0\notin\Omega_i^{r}\). Hence, the functions \(V_i^*(\chi_i)/\|\chi_i\|^2\) and \(\|\nabla_{\chi_i}V_i^*(\chi_i)\|/\|\chi_i\|\) are continuous on \(\Omega_i^{r}\). By the Weierstrass theorem, the constants \(\bar c_i:=\max_{\chi_i\in\Omega_i^{r}}V_i^*(\chi_i)/\|\chi_i\|^2<\infty\) and \(c_{\nabla i}:=\max_{\chi_i\in\Omega_i^{r}}\|\nabla_{\chi_i}V_i^*(\chi_i)\|/\|\chi_i\|<\infty\) exist, which directly give \(V_i^*(\chi_i)\le\bar c_i\|\chi_i\|^2\) and \(\|\nabla_{\chi_i}V_i^*(\chi_i)\|\le c_{\nabla i}\|\chi_i\|\). Moreover, positive definiteness of \(V_i^*\) and compactness of \(\Omega_i^{r}\) imply \(\underline c_{V_i}:=\min_{\chi_i\in\Omega_i^{r}}V_i^*(\chi_i)/\|\chi_i\|^2>0\). Choosing \(\underline{\alpha}_i(s):=\underline c_{V_i}s^2\in\mathcal K_{\infty}\) yields \(\underline{\alpha}_i(\|\chi_i\|)\le V_i^*(\chi_i)\), completing the proof.
\end{proof}

\begin{remark}[From a comparison condition to a constructive design]
\label{rem:comparison_to_design}
Theorem~\ref{thm:normalized_predefined_time} only guarantees an assigned
deadline once the required two-power Lyapunov inequality is available.
Existing fixed-time formation results such as
\cite{Gong2025SecureFormationFDI} impose this decay directly on the unknown
ideal value function, while
\cite{Kokotakis2025PredefinedTimeRL} does not construct it from a
graph-coupled HJI. Here, no such condition is assumed on \(V_i^\ast\).
Instead, \(Q_{ii}\) is designed from the measurable coordination error
\(\chi_i\) so that the HJI generates the required two-power terms, which
are recovered through integral learning and moving-window analysis to
synthesize the critic and closed-loop gains for the assigned time budgets.
\end{remark}



\begin{remark}[Cost-induced predefined-time structure]
\label{rem:cost_induced_predefined_time}
Motivated by the predefined-time HJB construction in \cite{Kokotakis2025PredefinedTimeRL}, the convergence-rate terms are embedded directly into the measurable state penalty of the local graphical game. Specifically,
\begin{equation}
Q_{ii}(\chi_i)
=
\chi_i^\top\mathbf Q_i\chi_i
+
\lambda_x\kappa_{i1}
\|\chi_i\|^{2\gamma_1}
+
\lambda_x\kappa_{i2}
\|\chi_i\|^{2\nu},
\label{eq:predefined_time_Qii_design}
\end{equation}
where \(\mathbf Q_i=\mathbf Q_i^\top>0\), \(\kappa_{i1},\kappa_{i2}>0\), and \(\lambda_x>0\). The scalar \(\lambda_x\) is retained explicitly because it will later be selected by inverse deadline design rather than fixed independently of the desired convergence time. This construction depends only on the measurable coordination error \(\chi_i\) and therefore does not introduce the unknown HJI value into the running cost. By Lemma~\ref{lem:local_value_bounds}, \(V_i^*(\chi_i)\le\bar c_i\|\chi_i\|^2\) for all \(\chi_i\in\Omega_i^r\). Hence,
\begin{equation}
\|\chi_i\|^{2\gamma_1}
\ge
\bar c_i^{-\gamma_1}
\bigl(V_i^*(\chi_i)\bigr)^{\gamma_1},
\qquad
\|\chi_i\|^{2\nu}
\ge
\bar c_i^{-\nu}
\bigl(V_i^*(\chi_i)\bigr)^\nu.
\label{eq:cost_to_value_predefined_terms}
\end{equation}
Define \(c_{i1}:=\kappa_{i1}\bar c_i^{-\gamma_1}\) and \(c_{i2}:=\kappa_{i2}\bar c_i^{-\nu}\). Therefore, whenever the HJI equation contributes \(-Q_{ii}(\chi_i)\) to \(\dot V_i^*\), it simultaneously produces the dissipative terms \(-\lambda_xc_{i1}(V_i^*)^{\gamma_1}-\lambda_xc_{i2}(V_i^*)^\nu\). The predefined-time structure is thus generated by the graphical-game cost instead of being postulated as an external property of the unknown value function. The exponent \(\gamma_1\) is subsequently lifted to \(\mu=(\gamma_1+1)/2\) by the derivative-limited moving-window inequality, yielding the same pair \((\mu,\nu)\) that appears in the critic dynamics and allowing Lemma~\ref{lem:Gamma_comparison} to be used for explicit deadline synthesis. To determine the specific values of \(\bar c_i\), please refer to~\cite{VuDoan2026PredefinedTimeResilientIRL}.
\end{remark}

For the secure input, use the saturation-compatible utility
\begin{equation}
\mathcal U_i(u_{ci})
=
2
\sum_{\ell=1}^{m}
\bar u_{i\ell}r_{i\ell}
\int_{0}^{u_{ci,\ell}}
\tanh^{-1}
\left(
\frac{s}{\bar u_{i\ell}}
\right)ds ,
\label{eq:sat_utility_hji}
\end{equation}
where \(\mathbf R_i=\operatorname{diag}(r_{i1},\ldots,r_{im})>0\). Hence,
\begin{equation}
\frac{\partial\mathcal U_i}{\partial u_{ci}}
=
2\bar{\mathbf U}_i\mathbf R_i
\tanh^{-1}
\left(
\bar{\mathbf U}_i^{-1}u_{ci}
\right).
\label{eq:utility_gradient_hji}
\end{equation}
The neighbor utility \(\mathcal U_{ij}(u_{cj})\) is defined similarly with a positive matrix \(\mathbf R_{ij}\).
The saddle-point policies are obtained directly from the stationarity conditions of the original Hamiltonian. From
\begin{equation}
\frac{\partial\bar{\mathcal H}_i}{\partial u_{ci}}
=
\mathbf G_{ii}^{\chi\,\top}(x_i)
\nabla V_i
+
\frac{\partial\mathcal U_i}{\partial u_{ci}}
=
0,
\label{eq:stationarity_control}
\end{equation}

One obtains
\(u_{ci}^\ast:=
-\bar{\mathbf U}_i
\tanh\!\left(
\frac{1}{2}
\mathbf R_i^{-1}
\bar{\mathbf U}_i^{-1}
\mathbf G_{ii}^{\chi\,\top}(x_i)
\nabla V_i^\ast
\right)\).
Thus, \(u_{ci}^\ast\in\mathbb U_i\) by construction. In particular,
\(|u_{ci,\ell}^\ast|<\bar u_{i\ell}\),
\(\ell=1,\ldots,m\), for every finite value gradient.

Similarly, the stationarity condition
\(\partial\bar{\mathcal H}_i/\partial\varpi_i
=
\mathbf B_i^{\chi\,\top}(x_i)\nabla V_i
-
2\gamma_i^2\mathbf T_{ii}\varpi_i
=
0\)
gives
\(\varpi_i^\ast:=
\frac{1}{2\gamma_i^2}
\mathbf T_{ii}^{-1}
\mathbf B_i^{\chi\,\top}(x_i)
\nabla V_i^\ast\).
For each \(j\in\mathcal N_i\), the corresponding neighboring adversarial policy is
\(\varpi_j^\ast:=
\frac{1}{2\gamma_i^2}
\mathbf T_{ij}^{-1}
\mathbf B_{ij}^{\chi\,\top}(x_j)
\nabla V_i^\ast\).
The worst-case disturbance and same-channel FDI signal are selected as
\(\omega_i^\ast:=\mathbf S_{\omega i}\varpi_i^\ast\) and
\(u_{ai}^\ast:=\mathbf S_{ui}\varpi_i^\ast\),
where \(\mathbf S_{\omega i}\) and \(\mathbf S_{ui}\) are constant selection matrices.

Substituting the above secure-control and adversarial saddle-point policies into the local HJI equation gives
\begin{equation}
\begin{aligned}
0
={}&
Q_{ii}(\chi_i)
+
\mathcal U_i(u_{ci}^\ast)
+
\sum_{j\in\mathcal N_i}
\mathcal U_{ij}(u_{cj})
\\
&+
\nabla V_i^{\ast\top}
\Big[
F_i^\chi(\xi_i)
+
\mathbf G_{ii}^\chi(x_i)u_{ci}^\ast
+
\mathbf G_{i\mathcal N}^\chi(x_{\mathcal N_i})
u_{c\mathcal N_i}
\Big]
\\
&+
\frac{1}{4\gamma_i^2}
\nabla V_i^{\ast\top}
\mathbf B_i^\chi(x_i)
\mathbf T_{ii}^{-1}
\mathbf B_i^{\chi\,\top}(x_i)
\nabla V_i^\ast
\\
&+
\frac{1}{4\gamma_i^2}
\sum_{j\in\mathcal N_i}
\nabla V_i^{\ast\top}
\mathbf B_{ij}^\chi(x_j)
\mathbf T_{ij}^{-1}
\mathbf B_{ij}^{\chi\,\top}(x_j)
\nabla V_i^\ast .
\end{aligned}
\label{eq:closed_cost_induced_hji}
\end{equation}

Along the ideal saddle trajectory, the HJI identity equivalently gives
\begin{equation}
\begin{aligned}
\dot V_i^\ast
={}&
-Q_{ii}(\chi_i)
-\mathcal U_i(u_{ci}^\ast)
-\sum_{j\in\mathcal N_i}
\mathcal U_{ij}(u_{cj})
\\
&+
\gamma_i^2
(\varpi_i^\ast)^\top
\mathbf T_{ii}\varpi_i^\ast
+
\gamma_i^2
\sum_{j\in\mathcal N_i}
(\varpi_j^\ast)^\top
\mathbf T_{ij}\varpi_j^\ast .
\end{aligned}
\label{eq:value_derivative_predefined_hji}
\end{equation}
Using \(\mathcal U_i(u_{ci}^\ast)\ge0\), \(\mathcal U_{ij}(u_{cj})\ge0\), \eqref{eq:predefined_time_Qii_design}, and \eqref{eq:cost_to_value_predefined_terms}, the state-cost contribution satisfies
\[
-Q_{ii}(\chi_i)
\le
-\chi_i^\top\mathbf Q_i\chi_i
-\lambda_xc_{i1}(V_i^\ast)^{\gamma_1}
-\lambda_xc_{i2}(V_i^\ast)^\nu.
\]
. The quadratic term \(-\chi_i^\top\mathbf Q_i\chi_i\) is retained because it also provides an additional local dissipation margin, whereas the mixed powers produce the low- and high-order terms used by the moving-window comparison.


The cost-induced HJI equation \eqref{eq:closed_cost_induced_hji} is a nonlinear first-order partial differential equation in \(V_i^\ast\). In general, solving this PDE analytically is intractable, especially because the drift \(F_i^\chi(\xi_i)\) is unknown and the local state \(\chi_i\) contains graph-coupled information. Therefore, instead of seeking an exact closed-form HJI solution, the value function is approximated by a critic neural network and learned through an integral Bellman--Isaacs identity over finite data windows, as developed in the next section.

\subsection{Integral Bellman--Isaacs Residual}

The drift \(F_i^\chi(\xi_i)\) is unknown, so \eqref{eq:closed_cost_induced_hji} cannot be solved directly. Integrating \eqref{eq:value_derivative_predefined_hji} over \([t-T_i,t]\) yields \(0=V_i^\ast(\chi_i(t))-V_i^\ast(\chi_i(t-T_i))+\int_{t-T_i}^{t}\big[Q_{ii}(\chi_i)+\mathcal U_i(u_{ci}^\ast)+\sum_{j\in\mathcal N_i}\mathcal U_{ij}(u_{cj})-\gamma_i^2(\varpi_i^\ast)^\top\mathbf T_{ii}\varpi_i^\ast-\gamma_i^2\sum_{j\in\mathcal N_i}(\varpi_j^\ast)^\top\mathbf T_{ij}\varpi_j^\ast\big]d\tau\). The value function is still an infinite-horizon HJI value, whereas the interval \([t-T_i,t]\) is used only to generate an integral learning identity.

\begin{assumption}[RBF approximation of the HJI solution]
\label{ass:value_regular}
On the compact trajectory-relevant set \(\Omega_i^r:=\{\chi_i\in\Omega_i:\|\chi_i\|\ge r_{i,-}\}\), the HJI solution \(V_i^\ast\) is positive definite and continuously differentiable. Moreover, there exist an ideal constant weight vector \(W_i^\ast\in\mathbb R^{N_i}\), a continuously differentiable RBF basis vector \(\phi_i(\chi_i)=\operatorname{col}\big(\phi_{i1}(\chi_i),\ldots,\phi_{iN_i}(\chi_i)\big)\), and an approximation error \(\varepsilon_i(\chi_i)\) such that \(V_i^\ast(\chi_i)=W_i^{\ast\top}\phi_i(\chi_i)+\varepsilon_i(\chi_i)\), \(\chi_i\in\Omega_i^r\). Each RBF basis can be chosen as \(\phi_{ik}(\chi_i)=\exp\big(-\|\chi_i-\zeta_{ik}\|^2/\sigma_{ik}^2\big)\), \(k=1,\ldots,N_i\), where \(\zeta_{ik}\) and \(\sigma_{ik}>0\) denote, respectively, the center and width of the \(k\)-th basis function. Furthermore, the approximation error and its gradient are bounded on \(\Omega_i^r\) by the Weierstrass theorem, namely \(|\varepsilon_i(\chi_i)|\le\bar\varepsilon_i\) and \(\|\nabla\varepsilon_i(\chi_i)\|\le\bar\varepsilon_{di}\), \(\chi_i\in\Omega_i^r\).
\end{assumption}

By Assumption~\ref{ass:value_regular}, the critic approximation is chosen as \(\hat V_i(\chi_i)=\hat W_i^\top\phi_i(\chi_i)\), with \(\nabla\hat V_i(\chi_i)=\nabla\phi_i^\top(\chi_i)\hat W_i\). The approximate policies are obtained by replacing the unavailable gradient \(\nabla V_i^\ast\) with \(\nabla\hat V_i\), yielding \(\hat u_{ci}=-\bar{\mathbf U}_i\tanh\big(\frac{1}{2}\mathbf R_i^{-1}\bar{\mathbf U}_i^{-1}\mathbf G_{ii}^{\chi\,\top}(x_i)\nabla\phi_i^\top\hat W_i\big)\), \(\hat\varpi_i=\frac{1}{2\gamma_i^2}\mathbf T_{ii}^{-1}\mathbf B_i^{\chi\,\top}(x_i)\nabla\phi_i^\top\hat W_i\), and \(\hat\varpi_j=\frac{1}{2\gamma_i^2}\mathbf T_{ij}^{-1}\mathbf B_{ij}^{\chi\,\top}(x_j)\nabla\phi_i^\top\hat W_i\).

Let \(\Delta\phi_i(t)=\phi_i(\chi_i(t))-\phi_i(\chi_i(t-T_i))\), \(T_i>0\). For compact notation, define the approximate integral running term as \(\hat r_i:=Q_{ii}(\chi_i)+\mathcal U_i(\hat u_{ci})+\sum_{j\in\mathcal N_i}\mathcal U_{ij}(\hat u_{cj})-\gamma_i^2\hat\varpi_i^\top\mathbf T_{ii}\hat\varpi_i-\gamma_i^2\sum_{j\in\mathcal N_i}\hat\varpi_j^\top\mathbf T_{ij}\hat\varpi_j\). Likewise, let \(r_i^*(\tau)\) denote the corresponding ideal running term obtained from \(\hat r_i(\tau)\) by replacing the learned policies \(\hat u_{ci}\), \(\hat u_{cj}\), \(\hat\varpi_i\), and \(\hat\varpi_j\) with their ideal saddle-point counterparts \(u_{ci}^*\), \(u_{cj}^*\), \(\varpi_i^*\), and \(\varpi_j^*\), respectively. Combining the preceding integral HJI identity with this approximate running term, the online Bellman--Isaacs residual is defined as \(\delta_i(t)=\hat W_i^\top\Delta\phi_i(t)+\int_{t-T_i}^{t}\hat r_i(\tau)d\tau\). This residual is the central data equation of the proposed IRL design. If \(\hat W_i=W_i^\ast\), the approximation error is zero, and the policies are exactly the saddle-point policies, then \(\delta_i(t)=0\) follows directly from the integral HJI identity. Hence, driving \(\delta_i(t)\) to zero is equivalent to enforcing the cost-induced HJI equation along the measured trajectory without requiring explicit knowledge of \(F_i^\chi(\xi_i)\).


To improve learning beyond the current data window, a finite experience stack is stored. Let \(\mathcal D_i:=\{(\Delta\phi_{ik},\rho_{ik})\}_{k=1}^{M_i}\), where \(\Delta\phi_{ik}:=\phi_i(\chi_i(t_k))-\phi_i(\chi_i(t_k-T_i))\) and \(\rho_{ik}:=\int_{t_k-T_i}^{t_k}\hat r_i(\tau)d\tau\). The corresponding recorded Bellman--Isaacs residual is \(\delta_{ik}:=\hat W_i^\top\Delta\phi_{ik}+\rho_{ik}\). The use of recorded data follows the experience-replay and concurrent-learning philosophy in adaptive optimal control. It replaces the need for persistent excitation by requiring only that the stored data contain enough independent information. 

For numerical conditioning and to avoid excessively large update increments, define the normalized regressors \(\psi_i:=\dfrac{\Delta\phi_i}{1+\Delta\phi_i^\top\Delta\phi_i}\) and \(\psi_{ik}:=\dfrac{\Delta\phi_{ik}}{1+\Delta\phi_{ik}^\top\Delta\phi_{ik}}\). The normalization does not change the zero of the Bellman--Isaacs residual; it only regularizes the learning direction used in the critic update.
Define \(\mu_i:=1+\Delta\phi_i^\top\Delta\phi_i\) and \(\mu_{ik}:=1+\Delta\phi_{ik}^\top\Delta\phi_{ik}\).

The following results provide the finite-window bounds required in the critic analysis.

\begin{lemma}[Finite-window boundedness]
\label{lem:bounded_learned_signals}
Suppose that \(\chi_i(\tau)\in\Omega_i^r\) for \(\tau\in[t-T_i,t]\) and that \(\hat W_i\) evolves according to \eqref{eq:irl_update_new}. Since \(\mu_i,\mu_{ik}\ge1\), the normalized regressors are bounded and the retained replay stack is finite. Hence, \(\dot{\hat W}_i\) is locally integrable and \(\hat W_i(t)=\hat W_i(t-T_i)+\int_{t-T_i}^{t}\dot{\hat W}_i(\tau)d\tau\) is finite on every finite window. Compactness of \(\Omega_i^r\), bounded \(\phi_i,\nabla\phi_i\), and continuity of the learned saddle policies then imply bounded learned signals and, consequently, a finite constant \(\bar r_{\Delta i}>0\) such that \(|\hat r_i(\tau)-r_i^\ast(\tau)|\le\bar r_{\Delta i}\) on \([t-T_i,t]\).
\end{lemma}

\begin{proof}
    See~[\cite{VuDoan2026PredefinedTimeResilientIRL} Lemma~2]
    
\end{proof}

Define \(\widetilde W_i:=W_i^\ast-\hat W_i\). Let \(\dot\chi_i^\ast(\tau)\) denote the ideal saddle vector field evaluated at the measured state \(\chi_i(\tau)\), whereas \(\dot\chi_i(\tau)\) is generated by the actual measured trajectory, and define the trajectory mismatch \(\eta_i(t):=\int_{t-T_i}^{t}\nabla V_i^{\ast\top}(\chi_i(\tau))[\dot\chi_i(\tau)-\dot\chi_i^\ast(\tau)]d\tau\). Since the ideal HJI identity gives \(\nabla V_i^{\ast\top}\dot\chi_i^\ast+r_i^\ast=0\), equivalently \(\eta_i(t)=V_i^\ast(\chi_i(t))-V_i^\ast(\chi_i(t-T_i))+\int_{t-T_i}^{t}r_i^\ast(\tau)d\tau\); hence \(\eta_i(t)=0\) when the measured evolution coincides with the ideal saddle evolution. Using the online Bellman--Isaacs residual defined above, \(\delta_i(t)=\hat W_i^\top\Delta\phi_i(t)+\int_{t-T_i}^{t}\hat r_i(\tau)d\tau\). With \(\Delta\epsilon_i(t):=\epsilon_i(\chi_i(t))-\epsilon_i(\chi_i(t-T_i))\), one obtains \(W_i^{\ast\top}\Delta\phi_i(t)+\Delta\epsilon_i(t)+\int_{t-T_i}^{t}r_i^\ast(\tau)d\tau-\eta_i(t)=0\), and therefore \(\delta_i(t)=-\widetilde W_i^\top\Delta\phi_i(t)+\epsilon_{H_i}(t)\), where \(\epsilon_{H_i}(t):=-\Delta\epsilon_i(t)+\int_{t-T_i}^{t}[\hat r_i(\tau)-r_i^\ast(\tau)]d\tau+\eta_i(t)\). Hence, with \(\psi_i:=\Delta\phi_i/\mu_i\) and \(\bar\epsilon_{H_i}:=\epsilon_{H_i}/\mu_i\), \(\delta_i/\mu_i=-\widetilde W_i^\top\psi_i+\bar\epsilon_{H_i}\); analogously, \(\delta_{ik}/\mu_{ik}=-\widetilde W_i^\top\psi_{ik}+\bar\epsilon_{H_{ik}}\), \(k=1,\ldots,M_i\). Compactness of \(\Omega_i^r\) and boundedness of the admissible and learned signals imply a finite \(\bar c_{\eta_i}>0\) such that \(|\eta_i(t)|\le T_i\bar c_{\eta_i}\). Since \(|\Delta\epsilon_i(t)|\le2\bar\epsilon_i\), \(|\hat r_i-r_i^\ast|\le\bar r_{\Delta i}\), and \(\mu_i,\mu_{ik}\ge1\), it follows that \(|\bar\epsilon_{H_i}(t)|\le2\bar\epsilon_i+T_i(\bar r_{\Delta i}+\bar c_{\eta_i})\). Because the replay stack is finite, define \(\epsilon_{H_im}:=\max\{2\bar\epsilon_i+T_i(\bar r_{\Delta i}+\bar c_{\eta_i}),\max_{1\le k\le M_i}|\bar\epsilon_{H_{ik}}|\}\); then \(|\bar\epsilon_{H_i}(t)|\le\epsilon_{H_im}\) and \(|\bar\epsilon_{H_{ik}}|\le\epsilon_{H_im}\), \(k=1,\ldots,M_i\).

\begin{remark}
Lemma~\ref{lem:bounded_learned_signals} and the subsequent derivation establish the bounded-data decomposition required for the critic analysis. Specifically, Lemma~4 guarantees that the finite-window and replay quantities remain well defined and bounded, while the following analysis rewrites the Bellman residual as a critic-weight-error term plus a uniformly bounded mismatch accounting for value-function approximation, learned-to-ideal policy deviation, and trajectory mismatch. This form allows all nonideal effects to be treated as bounded perturbations in the subsequent fixed-/predefined-time Lyapunov analysis.
\end{remark}

\subsection{Fixed-Time Critic Update}

To expose the learning structure. Consider the residual objective \(\mathcal E_i(\hat W_i):=\dfrac{|\delta_i|^{\gamma_1+1}}{(\gamma_1+1)\mu_i}+\dfrac{|\delta_i|^{\gamma_2+1}}{(\gamma_2+1)\mu_i}+\sum_{k=1}^{M_i}\left(\dfrac{|\delta_{ik}|^{\gamma_1+1}}{(\gamma_1+1)\mu_{ik}}+\dfrac{|\delta_{ik}|^{\gamma_2+1}}{(\gamma_2+1)\mu_{ik}}\right)\), where \(0<\gamma_1<1<\gamma_2\). The exponent \(\gamma_1\) increases the learning action near the residual origin, while \(\gamma_2\) prevents slow learning when the residual is large. Thus, the residual objective is shaped in the same two-power spirit as a fixed-time Lyapunov construction.

Along the approximate saddle-point trajectory induced by
\(\hat u_{ci}\), \(\hat\varpi_i\), and
\(\hat\varpi_j\), the local closed-loop
dynamics are
\(\dot\chi_i=
F_i^\chi(\xi_i)
+\mathbf G_{ii}^\chi(x_i)\hat u_{ci}
+\mathbf G_{i\mathcal N_i}^\chi(x_{\mathcal N_i})u_{c\mathcal N_i}
+\mathbf B_i^\chi(x_i)\hat\varpi_i
+\sum_{j\in\mathcal N_i}\mathbf B_{ij}^\chi(x_j)\hat\varpi_j\).
Define the corresponding approximate Hamiltonian as
\(\widehat{\mathcal H}_i=
\hat r_i+
\hat W_i^\top\nabla\phi_i(\chi_i)
\big[
F_i^\chi(\xi_i)
+\mathbf G_{ii}^\chi(x_i)\hat u_{ci}
+\mathbf G_{i\mathcal N_i}^\chi(x_{\mathcal N_i})u_{c\mathcal N_i}
+\mathbf B_i^\chi(x_i)\hat\varpi_i
+\sum_{j\in\mathcal N_i}\mathbf B_{ij}^\chi(x_j)\hat\varpi_j
\big]\).
By construction of the approximate saddle policies
\(\hat u_{ci}
=
-\bar{\mathbf U}_i
\tanh\!\big[
\frac{1}{2}
\mathbf R_i^{-1}
\bar{\mathbf U}_i^{-1}
\mathbf G_{ii}^{\chi\top}(x_i)
\nabla\phi_i^\top(\chi_i)\hat W_i
\big]\).
Hence,
\(\tanh^{-1}\!\big(
\bar{\mathbf U}_i^{-1}\hat u_{ci}
\big)
=
-\frac{1}{2}
\mathbf R_i^{-1}
\bar{\mathbf U}_i^{-1}
\mathbf G_{ii}^{\chi\top}(x_i)
\nabla\phi_i^\top(\chi_i)\hat W_i\).
Since
\(\partial\mathcal U_i(u_{ci})/\partial u_{ci}
=
2\bar{\mathbf U}_i\mathbf R_i
\tanh^{-1}\!\big(
\bar{\mathbf U}_i^{-1}u_{ci}
\big)\),
it follows that
\(\partial\mathcal U_i(\hat u_{ci})/\partial\hat u_{ci}
=
2\bar{\mathbf U}_i\mathbf R_i
\tanh^{-1}\!\big(
\bar{\mathbf U}_i^{-1}\hat u_{ci}
\big)
=
-\bar{\mathbf U}_i\mathbf R_i
\mathbf R_i^{-1}\bar{\mathbf U}_i^{-1}
\mathbf G_{ii}^{\chi\top}(x_i)
\nabla\phi_i^\top(\chi_i)\hat W_i
=
-\mathbf G_{ii}^{\chi\top}(x_i)
\nabla\phi_i^\top(\chi_i)\hat W_i\),
where the last equality follows since
\(\bar{\mathbf U}_i\) and \(\mathbf R_i\) are diagonal.
Therefore,
\(\mathbf G_{ii}^{\chi\top}(x_i)
\nabla\phi_i^\top(\chi_i)\hat W_i
+\partial\mathcal U_i(\hat u_{ci})/\partial\hat u_{ci}
=0\).
Likewise, from
\(\hat\varpi_i=
\frac{1}{2\gamma_i^2}
\mathbf T_{ii}^{-1}
\mathbf B_i^{\chi\top}(x_i)
\nabla\phi_i^\top(\chi_i)\hat W_i\),
one has
\(2\gamma_i^2\mathbf T_{ii}\hat\varpi_i
=
\mathbf B_i^{\chi\top}(x_i)
\nabla\phi_i^\top(\chi_i)\hat W_i\),
and hence
\(\mathbf B_i^{\chi\top}(x_i)
\nabla\phi_i^\top(\chi_i)\hat W_i
-2\gamma_i^2\mathbf T_{ii}\hat\varpi_i=0\).
Similarly, for every \(j\in\mathcal N_i\),
\(\hat\varpi_j=
\frac{1}{2\gamma_i^2}
\mathbf T_{ij}^{-1}
\mathbf B_{ij}^{\chi\top}(x_j)
\nabla\phi_i^\top(\chi_i)\hat W_i\)
implies
\(2\gamma_i^2\mathbf T_{ij}\hat\varpi_j
=
\mathbf B_{ij}^{\chi\top}(x_j)
\nabla\phi_i^\top(\chi_i)\hat W_i\),
so that
\(\mathbf B_{ij}^{\chi\top}(x_j)
\nabla\phi_i^\top(\chi_i)\hat W_i
-2\gamma_i^2\mathbf T_{ij}\hat\varpi_j=0\).

Hence, differentiating \(\widehat{\mathcal H}_i\) with respect to
\(\hat W_i\), while the neighboring secure-control coupling
\(u_{c\mathcal N_i}\) is fixed with respect to \(\hat W_i\), gives
\(\partial\widehat{\mathcal H}_i/\partial\hat W_i
=
\nabla\phi_i(\chi_i)
\big[
F_i^\chi(\xi_i)
+\mathbf G_{ii}^\chi(x_i)\hat u_{ci}
+\mathbf G_{i\mathcal N_i}^\chi(x_{\mathcal N_i})u_{c\mathcal N_i}
+\mathbf B_i^\chi(x_i)\hat\varpi_i
+\sum_{j\in\mathcal N_i}\mathbf B_{ij}^\chi(x_j)\hat\varpi_j
\big]
+
(\partial\hat u_{ci}/\partial\hat W_i)^\top
\big[
\mathbf G_{ii}^{\chi\top}(x_i)
\nabla\phi_i^\top(\chi_i)\hat W_i
+\partial\mathcal U_i(\hat u_{ci})/\partial\hat u_{ci}
\big]
+
(\partial\hat\varpi_i/\partial\hat W_i)^\top
\big[
\mathbf B_i^{\chi\top}(x_i)
\nabla\phi_i^\top(\chi_i)\hat W_i
-2\gamma_i^2\mathbf T_{ii}\hat\varpi_i
\big]
+
\sum_{j\in\mathcal N_i}
(\partial\hat\varpi_j/\partial\hat W_i)^\top
\big[
\mathbf B_{ij}^{\chi\top}(x_j)
\nabla\phi_i^\top(\chi_i)\hat W_i
-2\gamma_i^2\mathbf T_{ij}\hat\varpi_j
\big]\).
The three bracketed terms are therefore identically zero, and thus
\(\partial\widehat{\mathcal H}_i/\partial\hat W_i
=
\nabla\phi_i(\chi_i)
\big[
F_i^\chi(\xi_i)
+\mathbf G_{ii}^\chi(x_i)\hat u_{ci}
+\mathbf G_{i\mathcal N_i}^\chi(x_{\mathcal N_i})u_{c\mathcal N_i}
+\mathbf B_i^\chi(x_i)\hat\varpi_i
+\sum_{j\in\mathcal N_i}\mathbf B_{ij}^\chi(x_j)\hat\varpi_j
\big]
=
\nabla\phi_i(\chi_i)\dot\chi_i\).
Consequently,
\(\partial\delta_i/\partial\hat W_i
=
\int_{t-T_i}^{t}
\nabla\phi_i(\chi_i(\tau))\dot\chi_i(\tau)\,d\tau
=
\phi_i(\chi_i(t))
-\phi_i(\chi_i(t-T_i))
=
\Delta\phi_i\).

Moreover, since \((\Delta\phi_{ik},\rho_{ik})\) are fixed recorded quantities in \(\mathcal D_i\) during replay, \(\partial\delta_{ik}/\partial\hat W_i=\Delta\phi_{ik}\). Thus,
\begin{equation}
\frac{\partial\delta_i}{\partial\hat W_i}=\Delta\phi_i,
\qquad
\frac{\partial\delta_{ik}}{\partial\hat W_i}=\Delta\phi_{ik}.
\label{eq:irl_semigradient_new}
\end{equation}
We have
\begin{equation}
\begin{aligned}
\frac{\partial\mathcal E_i}{\partial\hat W_i}
={}&
\psi_i\!\left[
|\delta_i|^{\gamma_1}\operatorname{sgn}(\delta_i)
+
|\delta_i|^{\gamma_2}\operatorname{sgn}(\delta_i)
\right]
\\
&+
\sum_{k=1}^{M_i}
\psi_{ik}\!\left[
|\delta_{ik}|^{\gamma_1}\operatorname{sgn}(\delta_{ik})
+
|\delta_{ik}|^{\gamma_2}\operatorname{sgn}(\delta_{ik})
\right].
\end{aligned}
\label{eq:irl_objective_gradient_basic}
\end{equation}
Here, \(\psi_i=\Delta\phi_i/\mu_i\) and \(\psi_{ik}=\Delta\phi_{ik}/\mu_{ik}\). The learning gains are then inserted to tune the contribution of the current and recorded residuals. For compact notation, define



The two-power critic update with experience replay is selected as
\begin{equation}
\begin{aligned}
\dot{\hat W}_i
=
-\mathbf\Gamma_i
\Bigg\{
&\psi_i
\Big[
k_{i1}|\delta_i|^{\gamma_1}\operatorname{sgn}(\delta_i)
+
k_{i2}|\delta_i|^{\gamma_2}\operatorname{sgn}(\delta_i)
\Big]
\\
&+
\sum_{k=1}^{M_i}
\psi_{ik}
\Big[
k_{i1}^r|\delta_{ik}|^{\gamma_1}\operatorname{sgn}(\delta_{ik})
+
k_{i2}^r|\delta_{ik}|^{\gamma_2}\operatorname{sgn}(\delta_{ik})
\Big]
\Bigg\}.
\end{aligned}
\label{eq:irl_update_new}
\end{equation}
Here, \(\mathbf\Gamma_i=\mathbf\Gamma_i^\top>0\) and \(k_{i1},k_{i2},k_{i1}^r,k_{i2}^r>0\) are designer-selected learning gains that independently weight the low- and high-power corrections of the current and replayed residuals, respectively. These gains are introduced for learning-rate shaping rather than being inherited from the plant or HJI model. For the predefined-time analysis developed below, select \(\mathbf\Gamma_i=\alpha_{c,i}\mathbf I_{N_i}\), where \(\alpha_{c,i}>0\) is determined from the assigned critic-learning horizon \(T_W\).


\begin{remark}
For each follower \(i\), predefined-time critic convergence is not an independent physical objective, but is required because the learned policy \(\hat u_i\) depends directly on \(\nabla_{\chi_i}\hat V_i(\chi_i)\). Hence, the critic error \(\widetilde W_i\) enters the closed-loop analysis through the policy-approximation error. Driving \(\widetilde W_i\) into a bounded residual set within an assigned initial-condition-independent horizon prevents an arbitrarily long learning transient and supports the subsequent practical predefined-time stability result for the multi-agent system.
\end{remark}

\section{Stability Analysis}
\label{sec:stability_analysis}

\begingroup
\setlength{\abovedisplayskip}{3pt}
\setlength{\belowdisplayskip}{3pt}
\setlength{\abovedisplayshortskip}{2pt}
\setlength{\belowdisplayshortskip}{2pt}
\setlength{\jot}{1pt}

From the remark~\ref{rem:two_power_quantitative} recall
\(\mu=(\gamma_1+1)/2\) and \(\nu=(\gamma_2+1)/2\), where
\(0<\gamma_1<1<\gamma_2\), so that \(1/2<\mu<1<\nu\),
\(A_\gamma=(\gamma_2-1)/(\gamma_2-\gamma_1)\) and
\(B_\gamma=(1-\gamma_1)/(\gamma_2-\gamma_1)\), so that
\(A_\gamma,B_\gamma\in(0,1)\),
\(A_\gamma+B_\gamma=1\), and
\(\mu A_\gamma+\nu B_\gamma=1\).
The exact two-power constant appearing repeatedly below is
\(C_\gamma:=2\Gamma_{\!E}(A_\gamma)\Gamma_{\!E}(B_\gamma)/(\gamma_2-\gamma_1)\).


For follower \(i\), recall
\(\widetilde W_i:=W_i^\ast-\hat W_i\)
denote the critic weight-estimation error. From the data stack~\(\mathcal D_i\)~,let the retained informative replay matrix be
\(\bm\Psi_i:=[\,\bm\psi_{i1}\ \cdots\ \bm\psi_{iM_i}\,]\)
and define
\(\varsigma_i:=\sigma_{\min}(\bm\Psi_i)\).
Let \(T_{E,i}\) denote the first time at which the replay stack of follower \(i\) satisfies the finite-informativity condition (\(\varsigma_i>0\)), and set
\(T_E:=\max_{1\le i\le N}T_{E,i}\).
Hence, for every \(t\ge T_E\), all retained local replay stacks are informative. Here, \(T_E\) is a designer-induced data-informativity time, determined by the selected excitation mechanism and replay-stack construction; hence, by prescribing the excitation protocol and data-collection schedule, the designer determines the resulting value of \(T_E\) at which the finite-informativity condition is first achieved. In practice, the informative replay data may also be collected offline prior to closed-loop operation, in which case the replay stack is already informative at initialization and one may take \(T_E=0\). Nevertheless, for generality, the present analysis treats \(T_E\) as part of the online operational phase so that both pre-collected and online-acquired informative data are covered within the same framework.

\subsection{Predefined-Time Critic Learning}

\begin{assumption}[Retained replay informativity]
\label{ass:retained_replay}
For every follower \(i\), once \(t\ge T_{E,i}\), the full-rank replay stack
\(\bm\Psi_i\) is retained and satisfies
\(\varsigma_i=\sigma_{\min}(\bm\Psi_i)>0\).
\end{assumption}

To make the inverse time design explicit, select the critic gain of follower \(i\) as a scalar
\(\alpha_{c,i}>0\) and consider
\begin{equation}
V_{c,i}:=\frac{1}{2\alpha_{c,i}}\widetilde W_i^\top\widetilde W_i.
\label{eq:critic_local_Lyapunov}
\end{equation}
Consequently,
\(\|\widetilde W_i\|^2=2\alpha_{c,i}V_{c,i}\).
Throughout the analysis below, \(\mathbf\Gamma_i=\alpha_{c,i}\mathbf I_{N_i}\) as specified in \eqref{eq:irl_update_new}, and \(\delta_i\) and \(\delta_{ik}\) denote the normalized online and recorded residuals used in the critic update.

\begin{theorem}[Predefined-time integral critic]
\label{thm:critic_predefined_multi}
Suppose that Assumption~\ref{ass:retained_replay} holds. Let \(T_W>0\) be an assigned critic-learning horizon measured from the data-informativity instant \(T_E\), independently of the initial critic-weight errors.

\emph{(i) Exact residual case:}
Suppose that, for every follower \(i\),
\(\bar\epsilon_{H_i}(t)\equiv0\) and
\(\bar\epsilon_{H_{ik}}=0\), \(k=1,\ldots,M_i\). Let \(A_\gamma,B_\gamma,\gamma_1,\gamma_2\) be as defined in Lemma~\ref{lem:Gamma_comparison}~and~Remark~\ref{rem:two_power_quantitative}.
If
\begin{equation}
\alpha_{c,i}
\ge
\frac{\Gamma_{\!E}(A_\gamma)\Gamma_{\!E}(B_\gamma)}
{T_W(\gamma_2-\gamma_1)\varsigma_i^2
(k_{i1}^r)^{A_\gamma}(k_{i2}^r)^{B_\gamma}
M_i^{\frac{(1-\gamma_1)(1-\gamma_2)}
{2(\gamma_2-\gamma_1)}}},
\label{eq:alpha_ci_exact}
\end{equation}
then
\begin{equation}
\widetilde W_i(t)=0,
\qquad
t\ge T_E+T_W,
\qquad
i=1,\ldots,N.
\label{eq:multi_exact_critic_deadline}
\end{equation}

\emph{(ii) Nonzero residual case:}
For the known residual bound
\(\epsilon_{H_i m}>0\) defined above, select
\(\theta_W\in(0,1)\).
Define
\(\underline h_{i1}:=2^{-(\gamma_1+1)}
k_{i1}^r\varsigma_i^{\gamma_1+1}\) and
\(\underline h_{i2}:=2^{-(\gamma_2+1)}
k_{i2}^rM_i^{(1-\gamma_2)/2}\varsigma_i^{\gamma_2+1}\).
If
\begin{equation}
\alpha_{c,i}
\ge
\frac{\Gamma_{\!E}(A_\gamma)\Gamma_{\!E}(B_\gamma)}
{(1-\theta_W)T_W(\gamma_2-\gamma_1)}
\underline h_{i1}^{-A_\gamma}
\underline h_{i2}^{-B_\gamma},
\label{eq:alpha_ci_practical}
\end{equation}
then \(V_{c,i}\) enters, no later than \(T_E+T_W\), the forward-invariant set
\(\mathcal B_{W_i}:=\{V_{c,i}\in\mathbb R_{\ge0}:V_{c,i}\le\bar V_{c,i}\}\),
where \(\bar V_{c,i}\) is the unique nonnegative solution of
\begin{equation}
h_{i1}\bar V_{c,i}^{\mu}
+h_{i2}\bar V_{c,i}^{\nu}
=
\frac{\Delta_{W_i}}{\theta_W}.
\label{eq:Vbar_ci_root}
\end{equation}
Here,
\begin{equation}
h_{i1}:=\underline h_{i1}(2\alpha_{c,i})^\mu,
\qquad
h_{i2}:=\underline h_{i2}(2\alpha_{c,i})^\nu,
\label{eq:critic_hi_definition}
\end{equation}
and
\begin{equation}
\begin{aligned}
\Delta_{W_i}
:={}&
(1+2\,3^{\gamma_1})
(k_{i1}+M_i k_{i1}^r)
\epsilon_{H_i m}^{\gamma_1+1}
\\
&+
(1+2\,3^{\gamma_2})
(k_{i2}+M_i k_{i2}^r)
\epsilon_{H_i m}^{\gamma_2+1}.
\end{aligned}
\label{eq:Delta_Wi}
\end{equation}
Consequently,
\begin{equation}
\|\widetilde W_i(t)\|
\le
\bar W_{c,i}
:=
\sqrt{2\alpha_{c,i}\bar V_{c,i}},
\qquad
t\ge T_E+T_W.
\label{eq:local_post_learning_critic_ball}
\end{equation}

\end{theorem}

\begin{proof}
Since \(W_i^\ast\) is constant, \(\dot{\widetilde W}_i=-\dot{\hat W}_i\). Differentiating \eqref{eq:critic_local_Lyapunov} and substituting the critic update~\eqref{eq:irl_update_new} give
\(
\dot V_{c,i}
=
\widetilde W_i^\top\bm\psi_i(t)
[
k_{i1}\spow{\delta_i(t)}{\gamma_1}
+
k_{i2}\spow{\delta_i(t)}{\gamma_2}
]
+
\sum_{k=1}^{M_i}
\widetilde W_i^\top\bm\psi_{ik}
[
k_{i1}^r\spow{\delta_{ik}(t)}{\gamma_1}
+
k_{i2}^r\spow{\delta_{ik}(t)}{\gamma_2}
].
\)

In the exact-residual case, \(\delta_i(t)=-\widetilde W_i^\top\bm\psi_i(t)\) and \(\delta_{ik}(t)=-\widetilde W_i^\top\bm\psi_{ik}\). Hence,
\(
\dot V_{c,i}
=
-
k_{i1}|\widetilde W_i^\top\bm\psi_i(t)|^{\gamma_1+1}
-
k_{i2}|\widetilde W_i^\top\bm\psi_i(t)|^{\gamma_2+1}
-
k_{i1}^r\sum_{k=1}^{M_i}
|\widetilde W_i^\top\bm\psi_{ik}|^{\gamma_1+1}
-
k_{i2}^r\sum_{k=1}^{M_i}
|\widetilde W_i^\top\bm\psi_{ik}|^{\gamma_2+1}.
\)
The online terms are nonpositive and may be discarded. By the replay coercivity result,
\(
\dot V_{c,i}
\le
-k_{i1}^r\varsigma_i^{\gamma_1+1}
\|\widetilde W_i\|^{\gamma_1+1}
-
k_{i2}^rM_i^{\frac{1-\gamma_2}{2}}
\varsigma_i^{\gamma_2+1}
\|\widetilde W_i\|^{\gamma_2+1}.
\)
Using \(\|\widetilde W_i\|^2=2\alpha_{c,i}V_{c,i}\), this becomes
\(
\dot V_{c,i}
\le
-k_{i1}^r\varsigma_i^{\gamma_1+1}
(2\alpha_{c,i})^\mu
V_{c,i}^{\mu}
-
k_{i2}^rM_i^{\frac{1-\gamma_2}{2}}
\varsigma_i^{\gamma_2+1}
(2\alpha_{c,i})^\nu
V_{c,i}^{\nu}.
\)

Applying Lemma~\ref{lem:Gamma_comparison} gives
\(
T_{c,i}
<
\frac{
\Gamma_{\!E}(A_\gamma)
\Gamma_{\!E}(B_\gamma)
}{
\alpha_{c,i}
(\gamma_2-\gamma_1)
\varsigma_i^2
(k_{i1}^r)^{A_\gamma}
(k_{i2}^r)^{B_\gamma}
}
M_i^{
\frac{
(1-\gamma_1)(\gamma_2-1)
}{
2(\gamma_2-\gamma_1)
}
}.
\)
Indeed, \(\mu A_\gamma+\nu B_\gamma=1\) gives \((2\alpha_{c,i})^{\mu A_\gamma+\nu B_\gamma}=2\alpha_{c,i}\), while \((\gamma_1+1)A_\gamma+(\gamma_2+1)B_\gamma=2\). Therefore, condition \eqref{eq:alpha_ci_exact} guarantees \(T_{c,i}\le T_W\). Since \(T_{E,i}\le T_E\), \(T_{E,i}+T_{c,i}\le T_E+T_W\), which proves \eqref{eq:multi_exact_critic_deadline}.

\begin{lemma}[Perturbed signed-power inequality]
\label{lem:perturbed_power}
For every \(\gamma>0\) and \(z,\epsilon\in\mathbb R\), \(z\spow{-z+\epsilon}{\gamma}\le -2^{-(\gamma+1)}|z|^{\gamma+1}+\left(1+2\,3^\gamma\right)|\epsilon|^{\gamma+1}\). For \(\epsilon=0\), the sharper identity \(z\spow{-z}{\gamma}=-|z|^{\gamma+1}\) holds.
\end{lemma}

\begin{proof}
If \(|\epsilon|\le|z|/2\), then \(-z+\epsilon\) has the sign opposite to \(z\) and \(|-z+\epsilon|\ge|z|/2\), yielding the negative term above. If \(|\epsilon|>|z|/2\), then \(|z|<2|\epsilon|\) and \(|-z+\epsilon|<3|\epsilon|\), which yields the residual term. Combining the two cases proves the result.
\end{proof}

Consider now the nonzero-residual case. From the normalized residual relations established above, we have \(\bar\delta_i(t)=-\widetilde W_i^\top\bm\psi_i(t)+\bar\epsilon_{H_i}(t)\) and \(\bar\delta_{ik}(t)=-\widetilde W_i^\top\bm\psi_{ik}+\bar\epsilon_{H_{ik}}\), where \(\bar\delta_i(t):=\delta_i(t)/\mu_i(t)\) and \(\bar\delta_{ik}:=\delta_{ik}/\mu_{ik}\).

Applying Lemma~\ref{lem:perturbed_power} with \(\gamma=\gamma_1\) and \(\gamma=\gamma_2\) to every term in the preceding derivative expression, using \(|\bar\epsilon_{H_i}(t)|\le\epsilon_{H_i m}\) and \(|\bar\epsilon_{H_{ik}}|\le\epsilon_{H_i m}\), and discarding the remaining nonpositive online terms yield
\(
\dot V_{c,i}
\le
-
2^{-(\gamma_1+1)}
k_{i1}^r
\sum_{k=1}^{M_i}
|\widetilde W_i^\top\bm\psi_{ik}|^{\gamma_1+1}
-
2^{-(\gamma_2+1)}
k_{i2}^r
\sum_{k=1}^{M_i}
|\widetilde W_i^\top\bm\psi_{ik}|^{\gamma_2+1}
+
\Delta_{W_i}.
\)
Replay coercivity then gives
\(
\dot V_{c,i}
\le
-
2^{-(\gamma_1+1)}
k_{i1}^r
\varsigma_i^{\gamma_1+1}
\|\widetilde W_i\|^{\gamma_1+1}
-
2^{-(\gamma_2+1)}
k_{i2}^r
M_i^{\frac{1-\gamma_2}{2}}
\varsigma_i^{\gamma_2+1}
\|\widetilde W_i\|^{\gamma_2+1}
+
\Delta_{W_i}.
\)
Hence, using \(\|\widetilde W_i\|^2=2\alpha_{c,i}V_{c,i}\),
\begin{equation}
\dot V_{c,i}
\le
-h_{i1}V_{c,i}^{\mu}
-h_{i2}V_{c,i}^{\nu}
+\Delta_{W_i}.
\label{eq:Vci_practical_comparison}
\end{equation}

The left-hand side of \eqref{eq:Vbar_ci_root}, regarded as a function of \(\bar V_{c,i}\), is continuous, strictly increasing on \(\mathbb R_{\ge0}\), vanishes at zero, and diverges as \(\bar V_{c,i}\rightarrow\infty\). Hence, the root is unique. At \(V_{c,i}=\bar V_{c,i}\), one has \(\Delta_{W_i}=\theta_W(h_{i1}\bar V_{c,i}^{\mu}+h_{i2}\bar V_{c,i}^{\nu})\). For every \(V_{c,i}>\bar V_{c,i}\), strict monotonicity gives \(\Delta_{W_i}<\theta_W(h_{i1}V_{c,i}^{\mu}+h_{i2}V_{c,i}^{\nu})\), and therefore
\begin{equation}
\dot V_{c,i}
\le
-(1-\theta_W)
\left(
h_{i1}V_{c,i}^{\mu}
+
h_{i2}V_{c,i}^{\nu}
\right).
\label{eq:Vci_outside_BWi}
\end{equation}

Lemma~\ref{lem:Gamma_comparison} now yields
\(
T_{c,i,\mathcal B_{W_i}}
<
C_\gamma/
[
(1-\theta_W)
h_{i1}^{A_\gamma}
h_{i2}^{B_\gamma}
].
\)
Using \(\mu A_\gamma+\nu B_\gamma=1\) and \eqref{eq:critic_hi_definition}, one obtains
\(
h_{i1}^{A_\gamma}
h_{i2}^{B_\gamma}
=
2\alpha_{c,i}
\underline h_{i1}^{A_\gamma}
\underline h_{i2}^{B_\gamma}.
\)
Since \(C_\gamma=2\Gamma_{\!E}(A_\gamma)\Gamma_{\!E}(B_\gamma)/(\gamma_2-\gamma_1)\), it follows that
\(
T_{c,i,\mathcal B_{W_i}}
<
\frac{
\Gamma_{\!E}(A_\gamma)
\Gamma_{\!E}(B_\gamma)
}{
(1-\theta_W)
\alpha_{c,i}
(\gamma_2-\gamma_1)
\underline h_{i1}^{A_\gamma}
\underline h_{i2}^{B_\gamma}
},
\)~where \(T_{c,i,\mathcal B_{W_i}}\) denotes the first-entry time of
\(V_{c,i}(t)\) into the residual set \(\mathcal B_{W_i}\).
Condition \eqref{eq:alpha_ci_practical} therefore gives \(T_{c,i,\mathcal B_{W_i}}\le T_W\).

At the boundary \(V_{c,i}=\bar V_{c,i}\), \eqref{eq:Vci_practical_comparison} together with the defining relation \(\Delta_{W_i}=\theta_W(h_{i1}\bar V_{c,i}^{\mu}+h_{i2}\bar V_{c,i}^{\nu})\) implies \(\dot V_{c,i}\le-(1-\theta_W)\Delta_{W_i}/\theta_W\le0\). Hence, \(\mathcal B_{W_i}\) is forward invariant. Finally, \(\|\widetilde W_i(t)\|^2=2\alpha_{c,i}V_{c,i}(t)\le2\alpha_{c,i}\bar V_{c,i}\), which proves \eqref{eq:local_post_learning_critic_ball}.
\end{proof}

\begin{remark}[Assigned critic-learning horizon]
\label{rem:multi_critic_deadline}
The assigned quantity \(T_W\) is a design horizon measured after the data-informativity instant. The critic dynamics evolve from initialization, but the replay coercivity required by the proof is guaranteed only after \(T_E\). Thus, \(T_E\) is not an additional critic settling time generated by the learning law; it represents the finite time required for the measured data to become informative. Once this occurs, \(\alpha_{c,i}\) is related explicitly to the remaining learning horizon through \eqref{eq:alpha_ci_exact} or \eqref{eq:alpha_ci_practical}.
\end{remark}

\subsection{Predefined-Time Closed-Loop Stability}

The preceding theorem establishes the local learning result. We next propagate the critic error through the graph-coupled learned policies and combine it with the cost-induced HJI dissipation.


Define the stacked critic error \(\widetilde{\mathbf W}:=\operatorname{col}\{\widetilde W_1,\ldots,\widetilde W_N\}\) and, from Theorem~\ref{thm:critic_predefined_multi}, define \(\bar W_c:=(\sum_{i=1}^{N}\bar W_{c,i}^2)^{1/2}\). Then
\begin{equation}
\|\widetilde{\mathbf W}(t)\|
\le
\bar W_c,
\qquad
t\ge T_E+T_W.
\label{eq:global_post_learning_bound}
\end{equation}


For each local graphical game, define the actual-versus-reconstructed adversarial mismatches \(\Delta\varpi_{r,i}:=\varpi_i-\hat\varpi_i\) and \(\Delta\varpi_{r,ij}:=\varpi_j-\hat\varpi_j^{\,|i}\), \(j\in\mathcal N_i\), where \(\hat\varpi_j^{\,|i}\) denotes agent \(i\)'s learned reconstruction of the neighboring saddle-point adversarial policy \(\varpi_j^*\) defined above.


\begin{assumption}[Bounded local adversarial mismatch]
\label{ass:multi_adversarial_mismatch}
On the certified compact operating region, there exist known finite constants \(\bar\varpi_{r,i}\ge0\) and \(\bar\varpi_{r,ij}\ge0\) such that \(\|\Delta\varpi_{r,i}(t)\|\le\bar\varpi_{r,i}\) and \(\|\Delta\varpi_{r,ij}(t)\|\le\bar\varpi_{r,ij}\).
\end{assumption}

Let \(r_{\Omega_i}:=\max_{\xi\in\Omega_i^r}\|\xi\|\) and \(\bar p_{V_i}:=c_{\nabla i}r_{\Omega_i}\). Lemma~\ref{lem:local_value_bounds} gives \(\|\nabla V_i^\ast(\chi_i)\|\le c_{\nabla i}\|\chi_i\|\le\bar p_{V_i}\). By the RBF approximation, \(\|\nabla V_i^\ast-\nabla\hat V_i\|\le\bar\phi_{g,i}\|\widetilde W_i\|+\bar\varepsilon_{g,i}\), where \(\bar\phi_{g,i}:=\max_{\chi_i\in\Omega_i^r}\|\nabla\phi_i(\chi_i)\|\) and \(\bar\varepsilon_{g,i}:=\max_{\chi_i\in\Omega_i^r}\|\nabla\varepsilon_i(\chi_i)\|\).

Define \(\bar G_{ii}:=\max_{\Omega_i}\|\mathbf G_{ii}^{\chi}(x_i)\|\), \(\bar G_{i\mathcal N}:=\max_{\Omega_i}\|\mathbf G_{i\mathcal N}^{\chi}(x_{\mathcal N_i})\|\), \(\bar B_i:=\max_{\Omega_i}\|\mathbf B_i^\chi(x_i)\|\), and \(\bar B_{ij}:=\max_{\Omega_j}\|\mathbf B_{ij}^\chi(x_j)\|\). Since \(\tanh(\cdot)\) is globally one-Lipschitz, define

\begin{equation}
\begin{aligned}
\ell_{u,iW}
&:=
\frac12
\|\bar{\mathbf U}_i\|
\|
\mathbf R_i^{-1}
\bar{\mathbf U}_i^{-1}
\|
\bar G_{ii}\bar\phi_{g,i},
\\
\ell_{u,i0}
&:=
\frac12
\|\bar{\mathbf U}_i\|
\|
\mathbf R_i^{-1}
\bar{\mathbf U}_i^{-1}
\|
\bar G_{ii}\bar\varepsilon_{g,i}.
\end{aligned}
\label{eq:ell_ui_bounds}
\end{equation}

Then \(\|\hat u_{ci}-u_{ci}^\ast\|\le\ell_{u,iW}\|\widetilde W_i\|+\ell_{u,i0}\). For the neighboring secure-input stack, define \(\ell_{u,\mathcal N_iW}:=(\sum_{j\in\mathcal N_i}\ell_{u,jW}^2)^{1/2}\) and \(\ell_{u,\mathcal N_i0}:=(\sum_{j\in\mathcal N_i}\ell_{u,j0}^2)^{1/2}\). Consequently,
\[
\|
\hat u_{c\mathcal N_i}
-
u_{c\mathcal N_i}^\ast
\|
\le
\ell_{u,\mathcal N_iW}
\|
\widetilde{\mathbf W}
\|
+
\ell_{u,\mathcal N_i0}.
\]

Similarly, define
\[
\ell_{\varpi,iW}
:=
\frac{
\|\mathbf T_{ii}^{-1}\|
\bar B_i
\bar\phi_{g,i}
}{
2\gamma_i^2
},
\qquad
\ell_{\varpi,i0}
:=
\frac{
\|\mathbf T_{ii}^{-1}\|
\bar B_i
\bar\varepsilon_{g,i}
}{
2\gamma_i^2
},
\]
and, for \(j\in\mathcal N_i\),
\[
\ell_{\varpi,ijW}
:=
\frac{
\|\mathbf T_{ij}^{-1}\|
\bar B_{ij}
\bar\phi_{g,i}
}{
2\gamma_i^2
},
\qquad
\ell_{\varpi,ij0}
:=
\frac{
\|\mathbf T_{ij}^{-1}\|
\bar B_{ij}
\bar\varepsilon_{g,i}
}{
2\gamma_i^2
}.
\]
Using \(\varpi_i-\varpi_i^\ast=\hat\varpi_i-\varpi_i^\ast+\Delta\varpi_{r,i}\), one has \(\|\varpi_i-\varpi_i^\ast\|\le\ell_{\varpi,iW}\|\widetilde W_i\|+\ell_{\varpi,i0}+\bar\varpi_{r,i}\). Likewise, \(\|\varpi_j-\varpi_j^{\ast|i}\|\le\ell_{\varpi,ijW}\|\widetilde W_i\|+\ell_{\varpi,ij0}+\bar\varpi_{r,ij}\).

Collect the critic-dependent and critic-independent coefficients as
\[
\begin{aligned}
L_{W_i}
:={}&
\bar G_{ii}\ell_{u,iW}
+
\bar G_{i\mathcal N}\ell_{u,\mathcal N_iW}
+
\bar B_i\ell_{\varpi,iW}
+
\sum_{j\in\mathcal N_i}
\bar B_{ij}\ell_{\varpi,ijW},
\\
L_{0_i}
:={}&
\bar G_{ii}\ell_{u,i0}
+
\bar G_{i\mathcal N}\ell_{u,\mathcal N_i0}
+
\bar B_i
(
\ell_{\varpi,i0}
+
\bar\varpi_{r,i}
)
\\
&+
\sum_{j\in\mathcal N_i}
\bar B_{ij}
(
\ell_{\varpi,ij0}
+
\bar\varpi_{r,ij}
).
\end{aligned}
\]
Define \(c_{W_i}:=\bar p_{V_i}L_{W_i}\) and \(c_{0_i}:=\bar p_{V_i}L_{0_i}\). It follows that


\(\big|\nabla V_i^{\ast\top}\big[
\mathbf G_{ii}^{\chi}(\hat u_{ci}-u_{ci}^\ast)
+\mathbf G_{i\mathcal N}^{\chi}
(\hat u_{c\mathcal N_i}-u_{c\mathcal N_i}^\ast)
+\mathbf B_i^\chi(\varpi_i-\varpi_i^\ast)
+\sum_{j\in\mathcal N_i}
\mathbf B_{ij}^\chi
(\varpi_j-\varpi_j^{\ast|i})
\big]\big|
\le
c_{W_i}\|\widetilde{\mathbf W}\|+c_{0_i}\).


For every \(p>1\), \(b>0\), \(c\ge0\), and \(s\ge0\), the optimized Young inequality
\(cs\le bs^p+(p-1)p^{-p/(p-1)}b^{-1/(p-1)}c^{p/(p-1)}\)
will be used. For arbitrary \(b_{i1},b_{i2}>0\), define
\(\Delta_{x_i}:=
c_{0_i}
+\gamma_1(\gamma_1+1)^{-\frac{\gamma_1+1}{\gamma_1}}
b_{i1}^{-\frac1{\gamma_1}}
\left(\frac{c_{W_i}}{2}\right)^{\frac{\gamma_1+1}{\gamma_1}}
+\gamma_2(\gamma_2+1)^{-\frac{\gamma_2+1}{\gamma_2}}
b_{i2}^{-\frac1{\gamma_2}}
\left(\frac{c_{W_i}}{2}\right)^{\frac{\gamma_2+1}{\gamma_2}}\).

Splitting \(c_{W_i}\|\widetilde{\mathbf W}\|\) into two equal parts and applying the preceding optimized Young inequality with \(p=\gamma_1+1\) and \(p=\gamma_2+1\) yield
\(\big|\nabla V_i^{\ast\top}\big[\mathbf G_{ii}^{\chi}(\hat u_{ci}-u_{ci}^\ast)+\mathbf G_{i\mathcal N}^{\chi}(\hat u_{c\mathcal N_i}-u_{c\mathcal N_i}^\ast)+\mathbf B_i^\chi(\varpi_i-\varpi_i^\ast)+\sum_{j\in\mathcal N_i}\mathbf B_{ij}^\chi(\varpi_j-\varpi_j^{\ast|i})\big]\big|
\le b_{i1}\|\widetilde{\mathbf W}\|^{\gamma_1+1}
+b_{i2}\|\widetilde{\mathbf W}\|^{\gamma_2+1}
+\Delta_{x_i}\).

By the cost-induced predefined-time HJI construction of Section~\ref{Sec3}, the local ideal value along the actual learned closed-loop trajectory therefore satisfies

\begin{equation}
\begin{aligned}
\dot V_i^\ast
\le{}&
-\lambda_x c_{i1}(V_i^\ast)^{\gamma_1}
-\lambda_x c_{i2}(V_i^\ast)^{\nu}
+b_{i1}\|\widetilde{\mathbf W}\|^{\gamma_1+1}
\\
&\quad
+b_{i2}\|\widetilde{\mathbf W}\|^{\gamma_2+1}
+\Delta_{x_i}.
\end{aligned}
\label{eq:Vi_learned_predefined}
\end{equation}


For each follower, define the moving-window ideal-value energy
\begin{equation}
\mathcal I_i(t)
:=
\int_{t-T_i}^{t}V_i^\ast(\chi_i(\tau))\,d\tau,
\qquad
T_i>0.
\label{eq:local_value_window}
\end{equation}
Since \(\Omega_i^r\) is compact and \(V_i^\ast(\chi_i)\le\bar c_i\|\chi_i\|^2\), define \(\bar I_i:=T_i\bar c_i r_{\Omega_i}^2\), so that \(0\le\mathcal I_i(t)\le\bar I_i\). Boundedness of the post-learning closed-loop vector field on the compact operating region also implies the existence of \(L_{\dot V_i}>0\) such that
\begin{equation}
|\dot V_i^\ast(\chi_i(t))|
\le
L_{\dot V_i}.
\label{eq:local_Vdot_uniform_bound}
\end{equation}

\begin{lemma}[Derivative-limited low-order window bound~\cite{VuDoan2026FixedTimeMASIRL},~\cite{VuDoan2026PredefinedTimeResilientIRL}]
\label{lem:multi_low_window}
Let \(y:[t-T,t]\to\mathbb R_{\ge0}\) be absolutely continuous and satisfy \(|\dot y(\tau)|\le L_{\dot y}\) and \(\int_{t-T}^{t}y(\tau)d\tau\le\bar I\), where \(L_{\dot y},\bar I>0\). For \(0<\gamma_1<1\), one has
\(\int_{t-T}^{t}y^{\gamma_1}(\tau)d\tau
\ge
\rho_{\gamma_1,T}
\big(\int_{t-T}^{t}y(\tau)d\tau\big)^\mu\),
where
\(\rho_{\gamma_1,T}:=
\min\{2^{\gamma_1-1}L_{\dot y}^{-\frac{1-\gamma_1}{2}},
\,4^{\gamma_1-1}T^{1-\gamma_1}\bar I^{\frac{\gamma_1-1}{2}}\}\).
\end{lemma}

For \(y(\tau)=V_i^\ast(\chi_i(\tau))\), define
\(\rho_{i,T}:=
\min\{2^{\gamma_1-1}L_{\dot V_i}^{-\frac{1-\gamma_1}{2}},
\,4^{\gamma_1-1}T_i^{1-\gamma_1}\bar I_i^{\frac{\gamma_1-1}{2}}\}\).
Thus,
\(\int_{t-T_i}^{t}(V_i^\ast)^{\gamma_1}d\tau
\ge
\rho_{i,T}\mathcal I_i^\mu\).
Since \(\nu>1\), Jensen's inequality gives
\(\int_{t-T_i}^{t}(V_i^\ast)^\nu d\tau
\ge
T_i^{1-\nu}\mathcal I_i^\nu\).

Define \(T_{\max}:=\max_{1\le i\le N}T_i\) and \(t_0:=T_E+T_W+T_{\max}\). For every \(t\ge t_0\), the complete local interval \([t-T_i,t]\) lies in the post-learning critic regime for every follower. Hence,
\(\int_{t-T_i}^{t}\|\widetilde{\mathbf W}(\tau)\|^{\gamma_1+1}d\tau
\le
T_i\bar W_c^{\gamma_1+1}\)
and
\(\int_{t-T_i}^{t}\|\widetilde{\mathbf W}(\tau)\|^{\gamma_2+1}d\tau
\le
T_i\bar W_c^{\gamma_2+1}\).

Differentiating \eqref{eq:local_value_window} gives
\(\dot{\mathcal I}_i(t)
=
V_i^\ast(\chi_i(t))
-
V_i^\ast(\chi_i(t-T_i))
=
\int_{t-T_i}^{t}\dot V_i^\ast(\chi_i(\tau))d\tau\).
Define
\(a_i:=\lambda_xc_{i1}\rho_{i,T}\),
\(b_i:=\lambda_xc_{i2}T_i^{1-\nu}\), and
\(\Delta_{I_i}:=
T_i[b_{i1}\bar W_c^{\gamma_1+1}
+b_{i2}\bar W_c^{\gamma_2+1}
+\Delta_{x_i}]\).
Substitution of \eqref{eq:Vi_learned_predefined} into the preceding derivative identity, followed by the low-order moving-window bound, Jensen's high-order bound, and the post-learning critic-window bounds, gives
\(\dot{\mathcal I}_i
\le
-a_i\mathcal I_i^\mu
-b_i\mathcal I_i^\nu
+\Delta_{I_i}\)
for \(t\ge t_0\).


Define the aggregate moving-window value energy
\(\mathcal I_\chi:=\sum_{i=1}^{N}\mathcal I_i\).
Since \(0<\mu<1\),
\(\sum_{i=1}^{N}\mathcal I_i^\mu\ge\mathcal I_\chi^\mu\),
whereas, since \(\nu>1\),
\(\sum_{i=1}^{N}\mathcal I_i^\nu\ge N^{1-\nu}\mathcal I_\chi^\nu\).
Define
\(a_\chi:=\min_{1\le i\le N}a_i\),
\(b_\chi:=N^{1-\nu}\min_{1\le i\le N}b_i\), and
\(\Delta_I:=\sum_{i=1}^{N}\Delta_{I_i}\).
Then
\begin{equation}
\dot{\mathcal I}_\chi
\le
-a_\chi\mathcal I_\chi^\mu
-b_\chi\mathcal I_\chi^\nu
+\Delta_I.
\label{eq:Ichi_comparison}
\end{equation}

Define the aggregate critic Lyapunov function
\(V_c:=\sum_{i=1}^{N}V_{c,i}\).
Summing \eqref{eq:Vci_practical_comparison} gives
\(\dot V_c
\le
-\sum_{i=1}^{N}h_{i1}V_{c,i}^{\mu}
-\sum_{i=1}^{N}h_{i2}V_{c,i}^{\nu}
+\sum_{i=1}^{N}\Delta_{W_i}\).
Define
\(h_1:=\min_{1\le i\le N}h_{i1}\),
\(h_2:=N^{1-\nu}\min_{1\le i\le N}h_{i2}\), and
\(\Delta_W:=\sum_{i=1}^{N}\Delta_{W_i}\).
Using the same power-sum inequalities,
\begin{equation}
\dot V_c
\le
-h_1V_c^\mu
-h_2V_c^\nu
+\Delta_W.
\label{eq:network_Vc_comparison}
\end{equation}

To analyze the complete learned network, consider the mixed integral--pointwise Lyapunov-like functional
\begin{equation}
\mathcal J(t)
:=
\mathcal I_\chi(t)+V_c(t).
\label{eq:network_J}
\end{equation}

Adding \eqref{eq:Ichi_comparison} and \eqref{eq:network_Vc_comparison} yields
\(\dot{\mathcal J}\le-a_\chi\mathcal I_\chi^\mu-b_\chi\mathcal I_\chi^\nu-h_1V_c^\mu-h_2V_c^\nu+\Delta_I+\Delta_W\).
For \(0<\mu<1\), \(\mathcal I_\chi^\mu+V_c^\mu\ge(\mathcal I_\chi+V_c)^\mu=\mathcal J^\mu\), whereas
\(\mathcal I_\chi^\nu+V_c^\nu\ge2^{1-\nu}(\mathcal I_\chi+V_c)^\nu=2^{1-\nu}\mathcal J^\nu\)~(~See~\cite{VuDoan2026FixedTimeMASIRL}~).
Define \(C_1:=\min\{a_\chi,h_1\}\), \(C_2:=2^{1-\nu}\min\{b_\chi,h_2\}\), and
\(\Delta_J:=\Delta_I+\Delta_W\). Then
\(\dot{\mathcal J}\le-C_1\mathcal J^\mu-C_2\mathcal J^\nu+\Delta_J\) for \(t\ge t_0\).


Define the exact two-power integral constant
\(\mathfrak C_{\mu,\nu}:=
\dfrac{
\Gamma_{\!E}\!\left(\frac{\nu-1}{\nu-\mu}\right)
\Gamma_{\!E}\!\left(\frac{1-\mu}{\nu-\mu}\right)
}{
\nu-\mu
}\).
For a predefined-time guarantee measured from initialization, assume that the data-informativity time admits a known uniform bound \(T_E\le\bar T_E\). Let the designer assign \(T_p>0\) satisfying \(T_p>\bar T_E+T_{\max}\). Choose any \(\vartheta_T\in(0,1)\) and allocate
\(T_W:=\vartheta_T(T_p-\bar T_E-T_{\max})\) and
\(T_X:=(1-\vartheta_T)(T_p-\bar T_E-T_{\max})\), so that
\(\bar T_E+T_W+T_{\max}+T_X=T_p\).
Select \(\theta_J\in(0,1)\) and define
\(\Lambda_X:=\dfrac{\mathfrak C_{\mu,\nu}}{(1-\theta_J)T_X}\).

The state-side coefficients satisfy
\(a_\chi=\lambda_x\min_i\{c_{i1}\rho_{i,T}\}\) and
\(b_\chi=\lambda_xN^{1-\nu}\min_i\{c_{i2}T_i^{1-\nu}\}\).
Hence, the sufficient conditions \(a_\chi\ge\Lambda_X\) and
\(b_\chi\ge2^{\nu-1}\Lambda_X\) are enforced by
\(\lambda_x\ge
\max\left\{
\dfrac{\Lambda_X}{\min_i\{c_{i1}\rho_{i,T}\}},
\dfrac{(2N)^{\nu-1}\Lambda_X}{\min_i\{c_{i2}T_i^{1-\nu}\}}
\right\}\).

Likewise, the critic-side requirements \(h_1\ge\Lambda_X\) and
\(h_2\ge2^{\nu-1}\Lambda_X\) are guaranteed if, for every follower,
\(h_{i1}\ge\Lambda_X\) and
\(h_{i2}\ge(2N)^{\nu-1}\Lambda_X\).
Using \eqref{eq:critic_hi_definition}, define
\(\alpha_{X,i}:=
\dfrac12
\max\left\{
\left(\dfrac{\Lambda_X}{\underline h_{i1}}\right)^{1/\mu},
\left(\dfrac{(2N)^{\nu-1}\Lambda_X}{\underline h_{i2}}\right)^{1/\nu}
\right\}\).
Thus, in the practical residual case, choose
\(\alpha_{c,i}\ge\max\{\alpha_{W,i},\alpha_{X,i}\}\), where
\(\alpha_{W,i}:=
\dfrac{
\Gamma_{\!E}(A_\gamma)
\Gamma_{\!E}(B_\gamma)
}{
(1-\theta_W)T_W(\gamma_2-\gamma_1)
}
\underline h_{i1}^{-A_\gamma}
\underline h_{i2}^{-B_\gamma}\).

\begin{theorem}[Predefined-time practical stability]
\label{thm:closed_loop_predefined_multi}
Consider the nonlinear leader--follower system governed by the local
dynamics \eqref{eq:hji_local_dynamics_game}, under the
saturation-compatible learned secure policies and the predefined-time
critic updates. Suppose that
Assumptions~\ref{assk}--\ref{ass:retained_replay} and
Assumption~\ref{ass:multi_adversarial_mismatch} hold, that all closed-loop
trajectories remain inside the certified compact operating region, and
that the local ideal HJI values satisfy
Lemma~\ref{lem:local_value_bounds}. Suppose further that the
data-informativity time admits the known uniform bound
\(T_E\le\bar T_E\). Let the designer assign a deadline
\(T_p>\bar T_E+T_{\max}\), select \(T_W\) and \(T_X\) according to the
preceding deadline allocation, and choose
\(\theta_W,\theta_J\in(0,1)\). Select the critic gains
\(\alpha_{c,i}\), \(i=1,\ldots,N\), and the state-side gain
\(\lambda_x\) according to the corresponding gain conditions defined
above.

Then the mixed network functional \(\mathcal J(t)\) enters, no later than
the designer-assigned deadline \(T_p\), the compact forward-invariant set
\(\mathcal B_J:=\{\mathcal J\in\mathbb R_{\ge0}:
\mathcal J\le\bar{\mathcal J}\}\), where
\(\bar{\mathcal J}\) is the unique nonnegative solution of
\begin{equation}
C_1
\bar{\mathcal J}^{\mu}
+
C_2
\bar{\mathcal J}^{\nu}
=
\frac{\Delta_J}{\theta_J}.
\label{eq:network_Jbar_root}
\end{equation}

In particular, define
\(\bar V_i:=\max\{2L_{\dot V_i}T_i,2\bar{\mathcal J}/T_i\}\).
Then, for every \(t\ge T_p\),

In particular, define
\(r_{\chi_i}:=\max\{r_{i,-},\underline\alpha_i^{-1}(\bar V_i)\}\),
\(i=1,\ldots,N\).
Then, for every \(t\ge T_p\),
\(\|\chi_i(t)\|\le r_{\chi_i}\), \(i=1,\ldots,N\).

Consequently,
\(\|\chi(t)\|\le r_\chi:=(\sum_{i=1}^{N}r_{\chi_i}^2)^{1/2}\),
and the leader-rooted formation error satisfies
\begin{equation}
\|e(t)\|
\le
r_e
:=
\|\mathbf H^{-1}\otimes\mathbf I_n\|r_\chi,
\qquad
t\ge T_p.
\label{eq:e_predefined_bound}
\end{equation}
Moreover,
\begin{equation}
\|\widetilde W_i(t)\|
\le
r_{W_i}
:=
\min\!\left\{
\bar W_{c,i},
\sqrt{2\alpha_{c,i}\bar{\mathcal J}}
\right\},
\qquad
t\ge T_p.
\label{eq:Wi_final_bound}
\end{equation}
Thus, with
\(\mathbf\Gamma(t):=\operatorname{col}\{e(t),\widetilde W_1(t),\ldots,\widetilde W_N(t)\}\),
the learned closed loop is practically predefined-time stable with respect to
\begin{equation}
\begin{aligned}
\mathcal B_\Gamma
:=
\Big\{
&\operatorname{col}\{e,\widetilde W_1,\ldots,\widetilde W_N\}:
\|e\|\le r_e,\\
&\|\widetilde W_i\|\le r_{W_i},
\quad i=1,\ldots,N
\Big\}.
\end{aligned}
\label{eq:B_Gamma}
\end{equation}

If, in addition,
\(\bar\epsilon_{H_i}(t)\equiv0\) and
\(\bar\epsilon_{H_{ik}}=0\),
\(i=1,\ldots,N\), \(k=1,\ldots,M_i\), and
\(\bar\varepsilon_{g,i}=0\),
\(\bar\varpi_{r,i}=0\),
\(\bar\varpi_{r,ij}=0\)
for all admissible \(i\) and \(j\in\mathcal N_i\), define
\(r_{\chi,-}:=(\sum_{i=1}^{N}r_{i,-}^{2})^{1/2}\) and
\(r_{e,-}:=\|\mathbf H^{-1}\otimes\mathbf I_n\|r_{\chi,-}\).
Then \(\chi(t)\in\overline{\mathcal B}_{r_{\chi,-}}(0)\),
\(e(t)\in\overline{\mathcal B}_{r_{e,-}}(0)\), and
\(\widetilde W_i(t)=0\) for \(i=1,\ldots,N\) and all \(t\ge T_p\).

\end{theorem}

\begin{proof}
By Theorem~\ref{thm:critic_predefined_multi}, the critic-gain requirement specified above guarantees that every critic has entered its post-learning residual set no later than \(T_E+T_W\). Since \(T_i\le T_{\max}\), every moving interval \([t-T_i,t]\) lies entirely inside the post-learning critic regime for \(t\ge t_0=T_E+T_W+T_{\max}\).~By the definition of \(\alpha_{X,i}\) above, \(h_{i1}\ge\Lambda_X\) and \(h_{i2}\ge(2N)^{\nu-1}\Lambda_X\) for all \(i\). Hence, \(h_1=\min_i h_{i1}\ge\Lambda_X\) and \(h_2=N^{1-\nu}\min_i h_{i2}\ge2^{\nu-1}\Lambda_X\). Likewise, the preceding selection of \(\lambda_x\) gives \(a_\chi\ge\Lambda_X\) and \(b_\chi\ge2^{\nu-1}\Lambda_X\). Therefore, by the definitions of \(C_1\) and \(C_2\), \(C_1\ge\Lambda_X\) and \(C_2\ge\Lambda_X\).~The function \(s\mapsto C_1s^\mu+C_2s^\nu\) is continuous, strictly increasing on \(\mathbb R_{\ge0}\), vanishes at zero, and diverges as \(s\rightarrow\infty\). Hence, (61) admits a unique nonnegative solution. Whenever \(\mathcal J>\bar{\mathcal J}\), one has \(\Delta_J<\theta_J(C_1\mathcal J^\mu+C_2\mathcal J^\nu)\). Substitution into the preceding mixed-functional comparison inequality gives
\(\dot{\mathcal J}\le-(1-\theta_J)(C_1\mathcal J^\mu+C_2\mathcal J^\nu)\).~By the exact two-power settling-time integral, the time required after \(t_0\) for \(\mathcal J\) to enter \(\mathcal B_J\) satisfies
\(T_{\mathcal J}<\dfrac{\mathfrak C_{\mu,\nu}}{(1-\theta_J)C_1^{\frac{\nu-1}{\nu-\mu}}C_2^{\frac{1-\mu}{\nu-\mu}}}\).
Since
\((\nu-1)/(\nu-\mu)+(1-\mu)/(\nu-\mu)=1\)
and \(C_1,C_2\ge\Lambda_X\), one has
\(C_1^{\frac{\nu-1}{\nu-\mu}}C_2^{\frac{1-\mu}{\nu-\mu}}\ge\Lambda_X\).
Therefore,
\(T_{\mathcal J}\le\dfrac{\mathfrak C_{\mu,\nu}}{(1-\theta_J)\Lambda_X}=T_X\).
Because \(T_E\le\bar T_E\),
\(T_E+T_W+T_{\max}+T_{\mathcal J}
\le\bar T_E+T_W+T_{\max}+T_X=T_p\).
Hence, \(\mathcal J\) enters \(\mathcal B_J\) no later than the designer-assigned deadline \(T_p\).~At the boundary \(\mathcal J=\bar{\mathcal J}\), the preceding mixed-functional comparison inequality together with (61) gives
\(\dot{\mathcal J}\le-(1-\theta_J)\Delta_J/\theta_J\le0\).
Thus, \(\mathcal B_J\) is forward invariant.~It remains to recover pointwise bounds from the finite-window quantities. For every \(t\ge T_p\),
\(\mathcal I_i(t)\le\mathcal I_\chi(t)\le\mathcal J(t)\le\bar{\mathcal J}\).
If \(V_i^\ast(\chi_i(t))\le2L_{\dot V_i}T_i\), then the first alternative in the definition of \(\bar V_i\) already holds. Otherwise,
\(V_i^\ast(\chi_i(t))>2L_{\dot V_i}T_i\).
For every \(\tau\in[t-T_i,t]\), \eqref{eq:local_Vdot_uniform_bound} gives~\(V_i^\ast(\chi_i(\tau))
\ge V_i^\ast(\chi_i(t))
-L_{\dot V_i}|t-\tau|
\ge V_i^\ast(\chi_i(t))
-L_{\dot V_i}T_i
>\frac12V_i^\ast(\chi_i(t))\).
Thus,
\(\bar{\mathcal J}\ge\mathcal I_i(t)>(T_i/2)V_i^\ast(\chi_i(t))\),
which implies
\(V_i^\ast(\chi_i(t))<2\bar{\mathcal J}/T_i\).
Combining the two cases yields
\(V_i^\ast(\chi_i(t))
\le
\max\{2L_{\dot V_i}T_i,2\bar{\mathcal J}/T_i\}
=
\bar V_i\).~By Lemma~\ref{lem:local_value_bounds},
\(\underline\alpha_i(\|\chi_i(t)\|)\le V_i^\ast(\chi_i(t))\le\bar V_i\)
whenever \(\chi_i(t)\in\Omega_i^r\).
Since \(\underline\alpha_i\in\mathcal K_\infty\),
\(\|\chi_i(t)\|\le\underline\alpha_i^{-1}(\bar V_i)\)
on the nonterminal comparison region. If instead
\(\chi_i(t)=0\) or \(0<\|\chi_i(t)\|<r_{i,-}\), the corresponding
terminal condition of Remark~\ref{remark3} has already
been reached and no comparison estimate is required. Thus, with
\(r_{\chi_i}:=\max\{r_{i,-},
\underline\alpha_i^{-1}(\bar V_i)\}\),
one has \(\|\chi_i(t)\|\le r_{\chi_i}\), which establishes the
pointwise bound stated above. Stacking all local coordination errors gives
\(\|\chi(t)\|^2=\sum_{i=1}^{N}\|\chi_i(t)\|^2
\le\sum_{i=1}^{N}r_{\chi_i}^2\).~Finally,
\(\chi=(\mathbf H\otimes\mathbf I_n)e\).
Since the leader is globally reachable, \(\mathbf H\) is nonsingular. Therefore,
\(e=(\mathbf H^{-1}\otimes\mathbf I_n)\chi\), and
\(\|e(t)\|
\le
\|\mathbf H^{-1}\otimes\mathbf I_n\|\|\chi(t)\|\),
which proves the formation-error bound stated above. Moreover,
\(V_{c,i}(t)\le V_c(t)\le\mathcal J(t)\le\bar{\mathcal J}\),
and hence
\(\|\widetilde W_i(t)\|\le\sqrt{2\alpha_{c,i}\bar{\mathcal J}}\).
Combining this with the post-learning bound
\eqref{eq:local_post_learning_critic_ball}
gives the critic-error bound stated above.~Consider finally the exact case. Under the exact critic conditions in the theorem, Theorem~\ref{thm:critic_predefined_multi} gives
\(\widetilde W_i(t)=0\) for
\(t\ge T_E+T_W\), \(i=1,\ldots,N\).
Under
\(\bar\varepsilon_{g,i}=0\),
\(\bar\varpi_{r,i}=0\), and
\(\bar\varpi_{r,ij}=0\),
the learned secure policies coincide with their ideal saddle realizations and the actual adversarial channels coincide with the corresponding saddle channels after the critic deadline. Hence, \(\Delta_{W_i}=0\). Moreover, under the exact matching conditions, the original policy- and adversarial-mismatch terms vanish identically, so the Young-relaxed residual bound \(\Delta_{x_i}\) is no longer invoked. Consequently, the local state comparison reduces to its exact zero-residual form, and therefore \(\Delta_J=0\). For \(t\ge t_0\), while the local coordination errors remain in their nonterminal comparison regions, \(\dot{\mathcal J}\le-C_1\mathcal J^\mu-C_2\mathcal J^\nu\). The same inverse gain conditions imply that the associated exact two-power settling-time budget is no larger than \(T_X\). In view of Remark~\ref{remark3}, this comparison is required only until the corresponding terminal neighborhoods are reached. Consequently, in the exact-residual case, each local coordination error reaches \(\mathcal B_{r_{i,-}}(0)\) no later than \(T_p\), while \(\widetilde W_i(t)=0\) for all \(t\ge T_E+T_W\). Using \(e=(\mathbf H^{-1}\otimes\mathbf I_n)\chi\) together with the radii defined in the theorem gives the corresponding terminal bound for the leader-rooted formation error. This establishes the exact-residual terminal-region conclusion.
\end{proof}

\begin{remark}[Deadline accounting]
\label{rem:deadline_accounting}
The decomposition
\(T_p=\bar T_E+T_W+T_{\max}+T_X\)
is a conservative worst-case accounting used to certify the prescribed
deadline, rather than a sequential description of the physical closed-loop
evolution. The state and critic evolve simultaneously, and the formation
error may decrease substantially while the critic is still adapting.
Accordingly, actual convergence may occur before \(T_p\); the theorem only
guarantees that the certified residual set is reached no later than the
designer-assigned deadline.
\end{remark}

\subsection{Discussion}

The inner radius \(r_{i,-}>0\) represents a certification tradeoff rather than a physical limitation of the closed-loop system. Since the comparison bounds of Lemma~\ref{lem:local_value_bounds} are established on \(\{\chi_i\in\Omega_i^r:\|\chi_i\|\ge r_{i,-}\}\), the present analysis certifies predefined-time convergence to a prescribed neighborhood of the origin rather than to \(\chi_i=0\). Decreasing \(r_{i,-}\) tightens this neighborhood, but enlarges the computable bound \(\bar c_i\), which scales conservatively as \(r_{i,-}^{-2}\) (see~\cite{VuDoan2026PredefinedTimeResilientIRL}), and may require larger state-penalty or learning gains, increasing conservatism, control effort, and numerical sensitivity. If the trajectory enters \(\|\chi_i\|<r_{i,-}\), the outer-region comparison is no longer required because the prescribed formation objective has already been achieved. If convergence to the origin is additionally desired, a two-region design may switch inside \(\mathcal B_{r_{i,-}}(0)\) to a local stabilizing controller or a conventional ADP/IRL scheme without the outer-region two-power comparison.

\section{Simulation}

\FloatBarrier

\subsection{Simulation Setup}

Consider one leader and four nonlinear followers evolving in the two-dimensional plane. The state of follower \(i\) is \(x_i=\operatorname{col}(p_{ix},p_{iy},v_{ix},v_{iy})\in\mathbb R^4\), and the follower dynamics are \(\dot x_i=f_i(t,x_i)+\mathbf g_i(x_i)u_{ci}+\mathbf g_i(x_i)u_{ai}+\mathbf d_i\omega_i\), \(i\in\{1,\ldots,4\}\), where \(u_{ci}\in\mathbb R^2\) is the secure control input, \(u_{ai}\in\mathbb R^2\) is the actuator-side FDI signal, and \(\omega_i\in\mathbb R^2\) is the external disturbance. The unknown nonlinear drift is \(f_i(t,x_i)=\operatorname{col}(v_{ix},v_{iy},F_{ix},F_{iy})\), where \(F_{ix}=-0.30v_{ix}+0.16\sin(0.45p_{ix})+0.05\sin(p_{iy}v_{ix}/8)+0.035\cos(0.80t+\vartheta_i)+c_i\), \(F_{iy}=-0.28v_{iy}+0.15\sin(0.42p_{iy})+0.05\sin(p_{ix}v_{iy}/8)+0.035\sin(0.70t+\vartheta_i)-c_i\), \(c_i=0.02v_{ix}v_{iy}/[1+0.20(v_{ix}^2+v_{iy}^2)]\), and \(\vartheta_i=0.7(i-1)\). The state-dependent input matrix is \(\mathbf g_i(x_i)=\operatorname{col}(\mathbf 0_{2\times2},\operatorname{diag}(g_{ix},g_{iy}))\), where \(g_{ix}=1+0.08\cos[0.30p_{ix}+0.20(i-1)]\) and \(g_{iy}=0.95+0.08\sin[0.25p_{iy}-0.20(i-1)]\), while \(\mathbf d_i=\operatorname{col}(\mathbf 0_{2\times2},\mathbf I_2)\). The disturbance is \(\omega_i(t)=\operatorname{col}(\omega_{ix},\omega_{iy})\), where \(\omega_{ix}=e^{-\lambda_\omega t}\{0.045\sin([1.30+0.10(i-1)]t)+0.018\cos([2.20+0.07(i-1)]t+0.30(i-1))\}\) and \(\omega_{iy}=e^{-\lambda_\omega t}\{0.040\cos([1.50+0.08(i-1)]t)+0.020\sin([2.45+0.09(i-1)]t+0.20(i-1))\}\), with \(\lambda_\omega=0.06\). The FDI attack is activated at \(t_a=8\,\mathrm{s}\), with \(u_{ai}(t)=\mathbf 0_2\) for \(t<t_a\) and \(u_{ai}(t)=\operatorname{col}(u_{ai,x},u_{ai,y})\) for \(t\ge t_a\). Specifically, \(u_{ai,x}=e^{-\lambda_a(t-t_a)}[a_{xi}^{(1)}\sin(\omega_{xi}^{(1)}t+\vartheta_i)+a_{xi}^{(2)}\cos(\omega_{xi}^{(2)}t+\varphi_i)]\) and \(u_{ai,y}=e^{-\lambda_a(t-t_a)}[a_{yi}^{(1)}\cos(\omega_{yi}^{(1)}t+\varphi_i)+a_{yi}^{(2)}\sin(\omega_{yi}^{(2)}t+\vartheta_i)]\), where \(\lambda_a=0.08\), \(\vartheta_i=0.7(i-1)\), and \(\varphi_i=0.45+0.30(i-1)\). The heterogeneous attack amplitudes are \(a_x^{(1)}=\operatorname{col}(0.32,0.46,0.39,0.54)\), \(a_x^{(2)}=\operatorname{col}(0.15,0.11,0.18,0.13)\), \(a_y^{(1)}=\operatorname{col}(0.41,0.34,0.52,0.47)\), and \(a_y^{(2)}=\operatorname{col}(0.12,0.19,0.14,0.17)\), whereas the corresponding frequency vectors are \(\omega_x^{(1)}=\operatorname{col}(0.83,1.07,0.76,1.23)\), \(\omega_x^{(2)}=\operatorname{col}(1.64,1.41,1.89,1.52)\), \(\omega_y^{(1)}=\operatorname{col}(0.91,1.18,0.87,1.31)\), and \(\omega_y^{(2)}=\operatorname{col}(1.72,1.55,1.96,1.63)\).

The leader state is \(x_0=\operatorname{col}(p_{0x},p_{0y},v_{0x},v_{0y})\), with \(\dot p_{0x}=v_{0x}\), \(\dot p_{0y}=v_{0y}\), \(\dot v_{0x}=-0.035p_{0x}-0.160v_{0x}+0.480\cos(0.32t)\), and \(\dot v_{0y}=-0.030p_{0y}-0.180v_{0y}+0.420\sin(0.28t)\). The desired offsets are \(h_1=\operatorname{col}(-10,-6,0,0)\), \(h_2=\operatorname{col}(10,-6,0,0)\), \(h_3=\operatorname{col}(-10,6,0,0)\), and \(h_4=\operatorname{col}(10,6,0,0)\). The model is used solely to generate simulation data and is not required by the proposed learning controller. \(\mathbf A\) is the adjacency matrix of the cycle graph \(\mathcal C_4\), i.e., \(a_{12}=a_{21}=a_{23}=a_{32}=a_{34}=a_{43}=a_{41}=a_{14}=1\), with all remaining entries equal to zero. The leader-pinning matrix is \(\mathbf B_0=\operatorname{diag}(1.4,0,0,1.2)\), where \(\mathbf L\) denotes the Laplacian matrix and \(\mathbf H=\mathbf L+\mathbf B_0\).

\FloatBarrier
For each follower, the critic is implemented by a six-node Gaussian RBFNN, \(\hat V_i(\chi_i)=\hat W_i^\top\phi_i(\chi_i)\), where \(\hat W_i\in\mathbb R^6\) and \(\chi_i=\operatorname{col}(\chi_{pi},\chi_{vi})\in\mathbb R^4\). For numerical conditioning, the critic input is normalized as \(\bar\chi_i=\operatorname{col}(\chi_{pi}/5,\chi_{vi}/3)\), and \(\phi_i(\chi_i)=\operatorname{col}(\phi_{i1},\ldots,\phi_{i6})\), with \(\phi_{ik}(\chi_i)=\exp(-\|\bar\chi_i-c_{ik}\|^2/\sigma_{ik}^2)\). The Gaussian centers are distributed over the normalized operating region with overlapping widths. The initial conditions are \(x_0(0)=\operatorname{col}(0,0,1.2,1.6)\), \(x_1(0)=\operatorname{col}(-5,-10,2.2,0.8)\), \(x_2(0)=\operatorname{col}(6,-1,0.5,2.5)\), \(x_3(0)=\operatorname{col}(-4,9,2.0,2.3)\), and \(x_4(0)=\operatorname{col}(5,2,0.3,1.0)\). To evaluate the initial-condition-independent convergence property, the initial formation-error perturbations are additionally scaled by \(s\in\{0.5,1.0,1.5,2.0,2.5\}\). The simulation is conducted over \(t\in[0,60]\,\mathrm{s}\) using a fourth-order Runge--Kutta method with \(T_s=0.01\,\mathrm{s}\). The integral Bellman window is \(T_i=0.25\,\mathrm{s}\). Replay data are sampled every \(0.05\,\mathrm{s}\) over the pre-attack interval \(t\in[0,8]\,\mathrm{s}\), with at most \(80\) samples retained thereafter to preserve the required initial excitation. The critic is initialized at \(\hat W_i(0)=50\mathbf 1_6\) (see our prior work~\cite{VuDoan2026FixedTimeResilientIRL}, Sec.~V, for a systematic initial-weight selection procedure), with \(\gamma_1=0.65\), \(\gamma_2=1.70\), \(k_{i1}=1.20\), \(k_{i2}=0.18\), \(k_{i1}^{r}=1.20\), and \(k_{i2}^{r}=0.18\). The input penalty and actuator bound are \(\mathbf R_i=0.08\mathbf I_2\) and \(\bar{\mathbf U}_i=40\mathbf I_2\), respectively.

\FloatBarrier

\subsection{Simulation Results}
\label{sec:simulation_results}



\begin{figure}[H]
    \centering
    \includegraphics[width=\columnwidth]
    {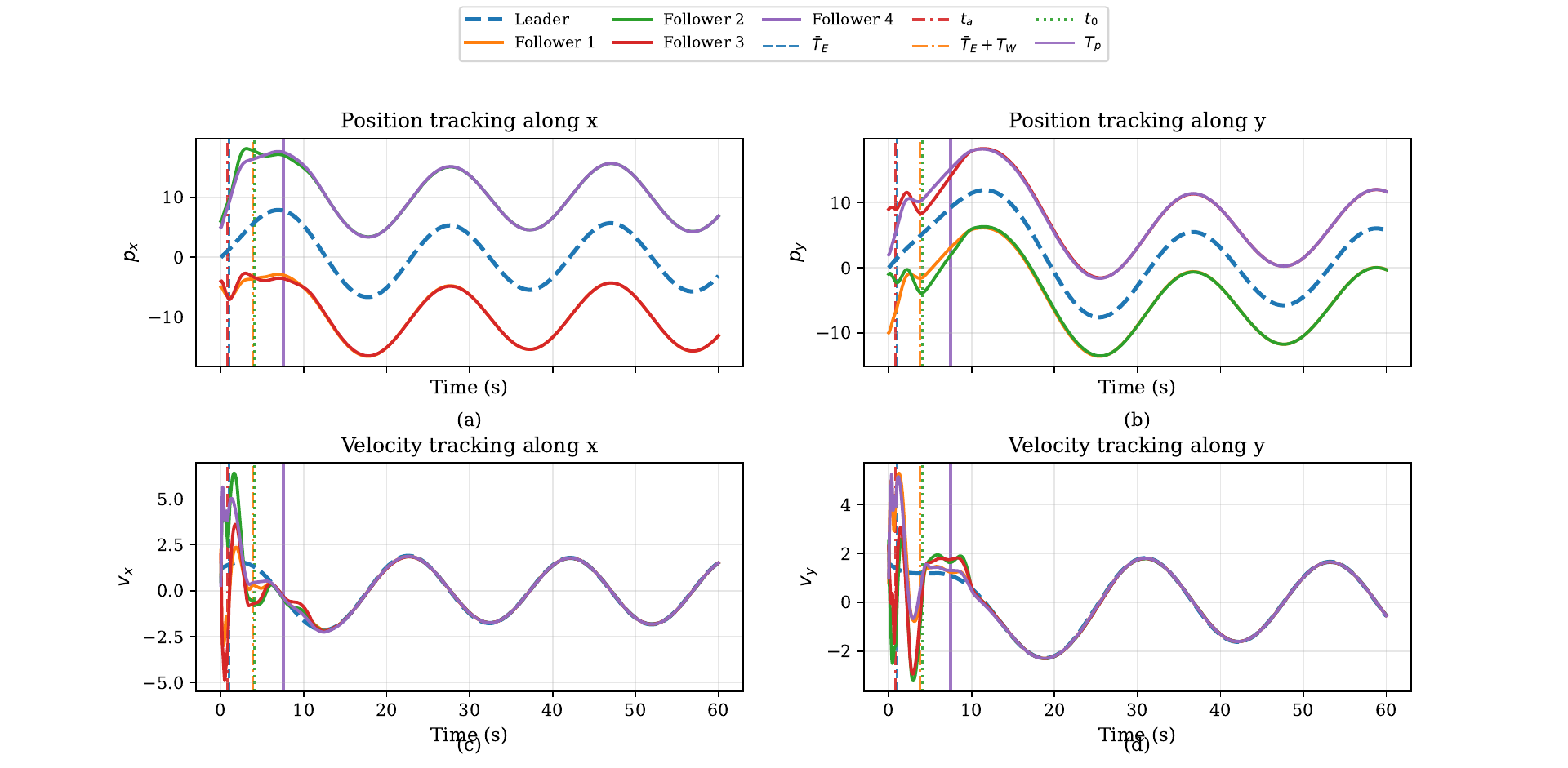}
    \caption{Position and velocity trajectories of the leader and
    followers together with the predefined-time milestones.}
    \label{fig:state_tracking}
\end{figure}




\begin{figure}[H]
    \centering
    \includegraphics[
        width=0.95\columnwidth,
        keepaspectratio
    ]{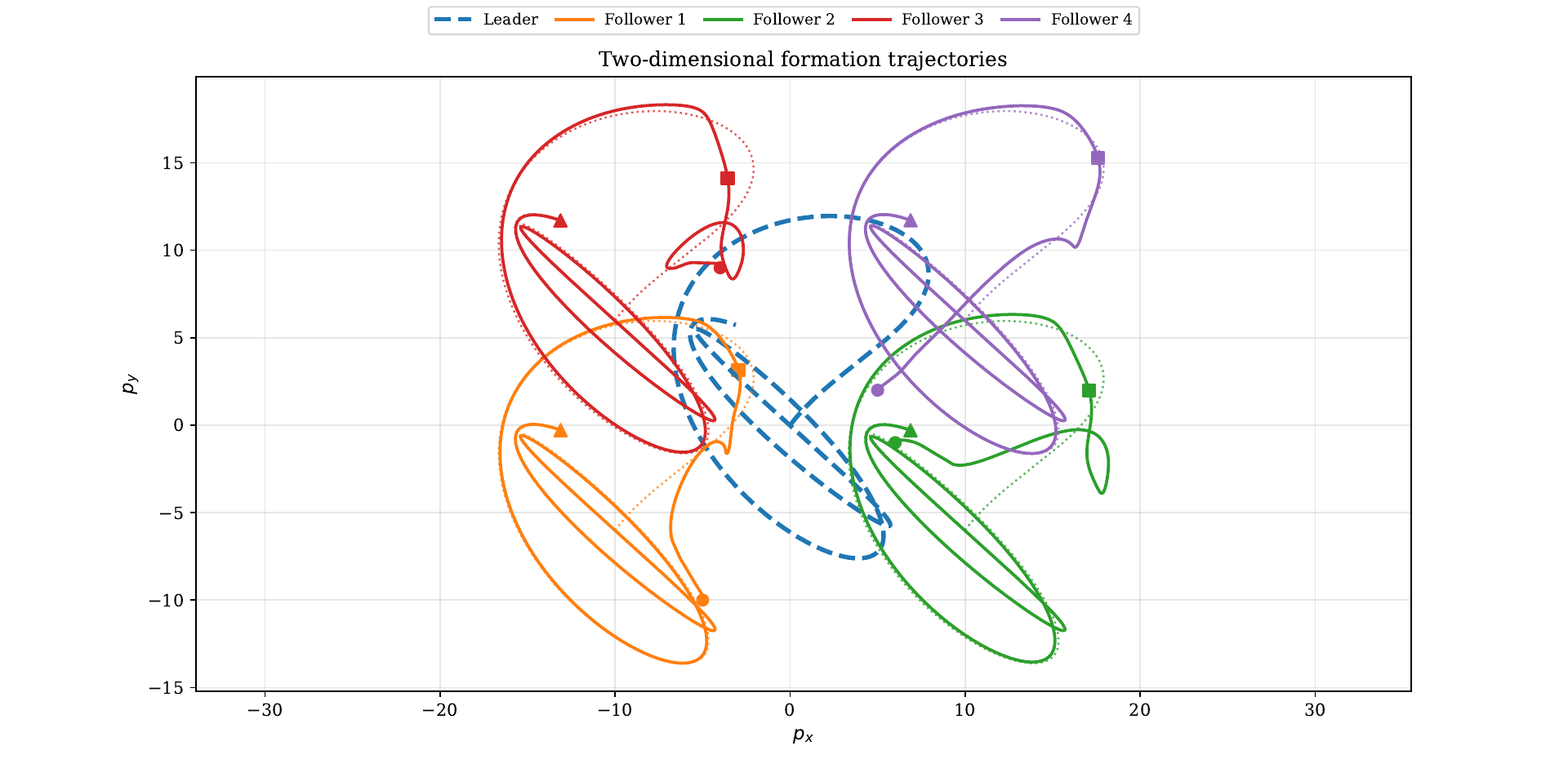}
    \caption{Two-dimensional trajectories of the leader and four
    followers; dotted curves denote the desired follower trajectories.}
    \label{fig:formation_2d}
\end{figure}

Figures~\ref{fig:state_tracking} and~\ref{fig:formation_2d} jointly illustrate the closed-loop formation behavior in both state and physical configuration spaces. As the assigned deadline \(T_p\) is approached, the follower velocities become nearly synchronized with the leader and the formation errors enter a small neighborhood of the origin while the prescribed relative offsets are maintained. This is consistent with the practical predefined-time result, which guarantees entrance into a residual ball rather than exact zero convergence under nonideal conditions. Its radius depends on the approximation residuals and the admissible disturbance/attack bounds. The zero-sum game represents a mathematically defined worst-case interaction; actual exogenous disturbances need not coincide exactly with the corresponding saddle realization. Hence, the small nonzero residual is both theoretically expected and physically meaningful.

The trajectories also expose the deadline--effort tradeoff: reducing \(T_p\) requires stronger closed-loop dissipation and larger state-side gains, yielding faster convergence but greater control effort and sensitivity to actuator limits. Thus, \(T_p\) should be regarded not only as a convergence specification but also as a design variable constrained by available actuation and the disturbance environment.


\begin{figure}[H]
    \centering
    \includegraphics[width=\columnwidth]
    {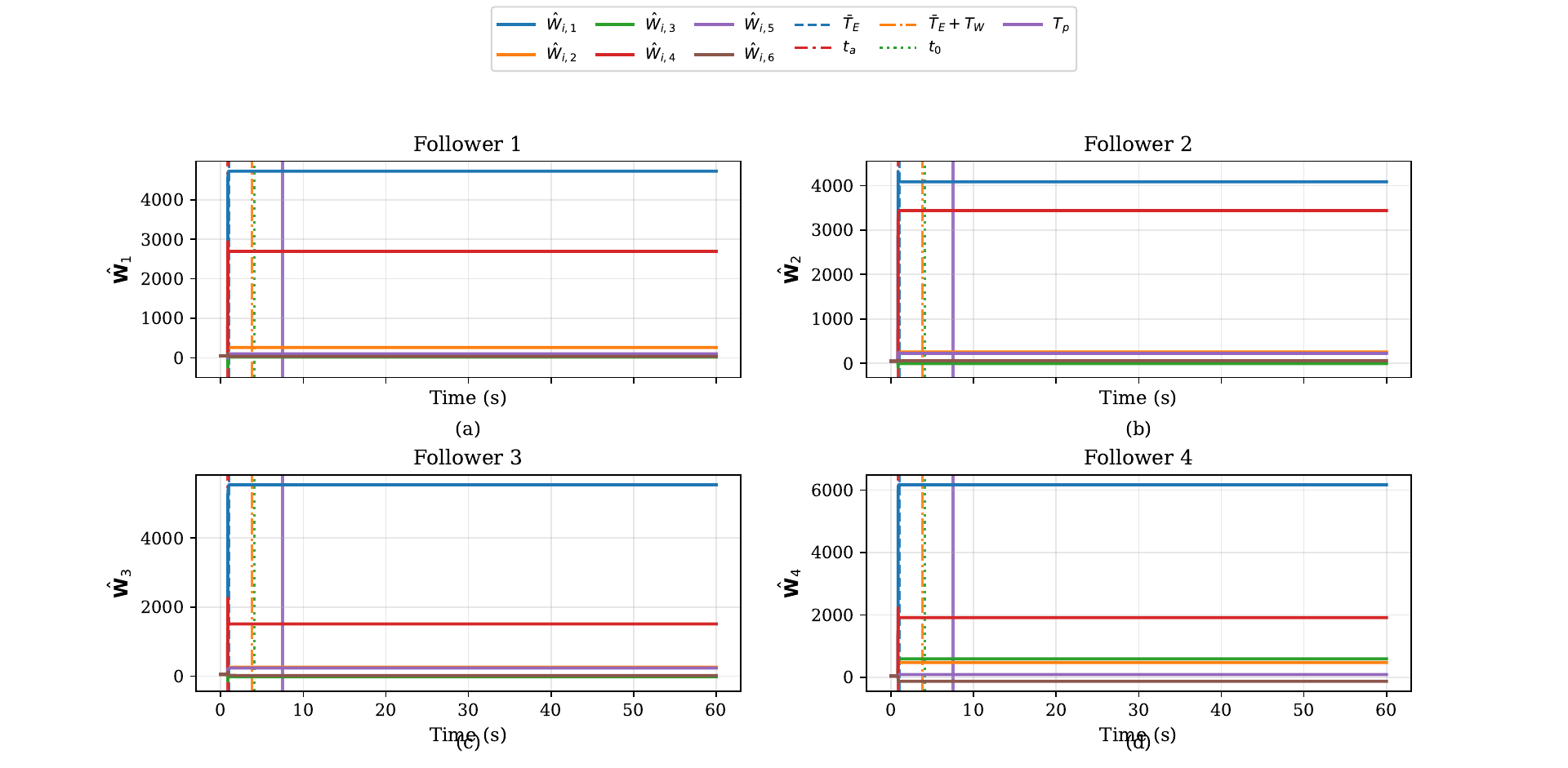}
    \caption{Critic-weight trajectories of the four followers together
    with the predefined-time learning milestones.}
    \label{fig:critic_weights}
\end{figure}

Figure~\ref{fig:critic_weights} shows a very short critic transient. Once the
informative replay stack becomes available, the dominant weight components
rapidly reach an essentially stationary bounded regime. This behavior results
from the proposed learning architecture rather than any imposed weight-freezing
mechanism. The retained replay stack preserves an algebraic information source
after the finite excitation interval, so the recorded Bellman--Isaacs residuals
continue to contribute to the update. The two-power law combines a high-order
correction for large critic errors with a low-order correction near the residual
set, while the critic gain \(\alpha_{c,i}\) is synthesized inversely from the
assigned horizon \(T_W\); hence, a shorter learning budget produces a more
aggressive initial adaptation rate. The different steady weight magnitudes,
including dominant components of order \(10^3\), are expected because the agents
have different local regressors and value-function approximations, and their
numerical equality is neither required nor implied by the theory. Quantitatively,
the visible adaptation is concentrated within approximately the first
\(5\!-\!7\,\mathrm{s}\), after which the retained weights remain essentially
stationary over the remainder of the \(60\,\mathrm{s}\) simulation. The dominant
components settle at approximately \(4\times10^3\) to \(6\times10^3\) across the
four followers, while the remaining components converge to substantially smaller
levels. Moreover, no appreciable post-learning drift is observed after the
assigned critic-learning interval, providing numerical support for the bounded
finite-window/replay analysis and the prescribed learning deadline.

This behavior also reveals an implementation tradeoff: reducing \(T_W\)
accelerates learning but increases the critic gain, parameter excursion, and
numerical stiffness. Thus, a short prescribed horizon should be balanced against
basis conditioning, sampling rate, and finite-precision implementation,
motivating conditioning-aware replay selection and gain regularization for
larger critic architectures.


\begin{figure}[H]
    \centering
    \includegraphics[width=\columnwidth]
    {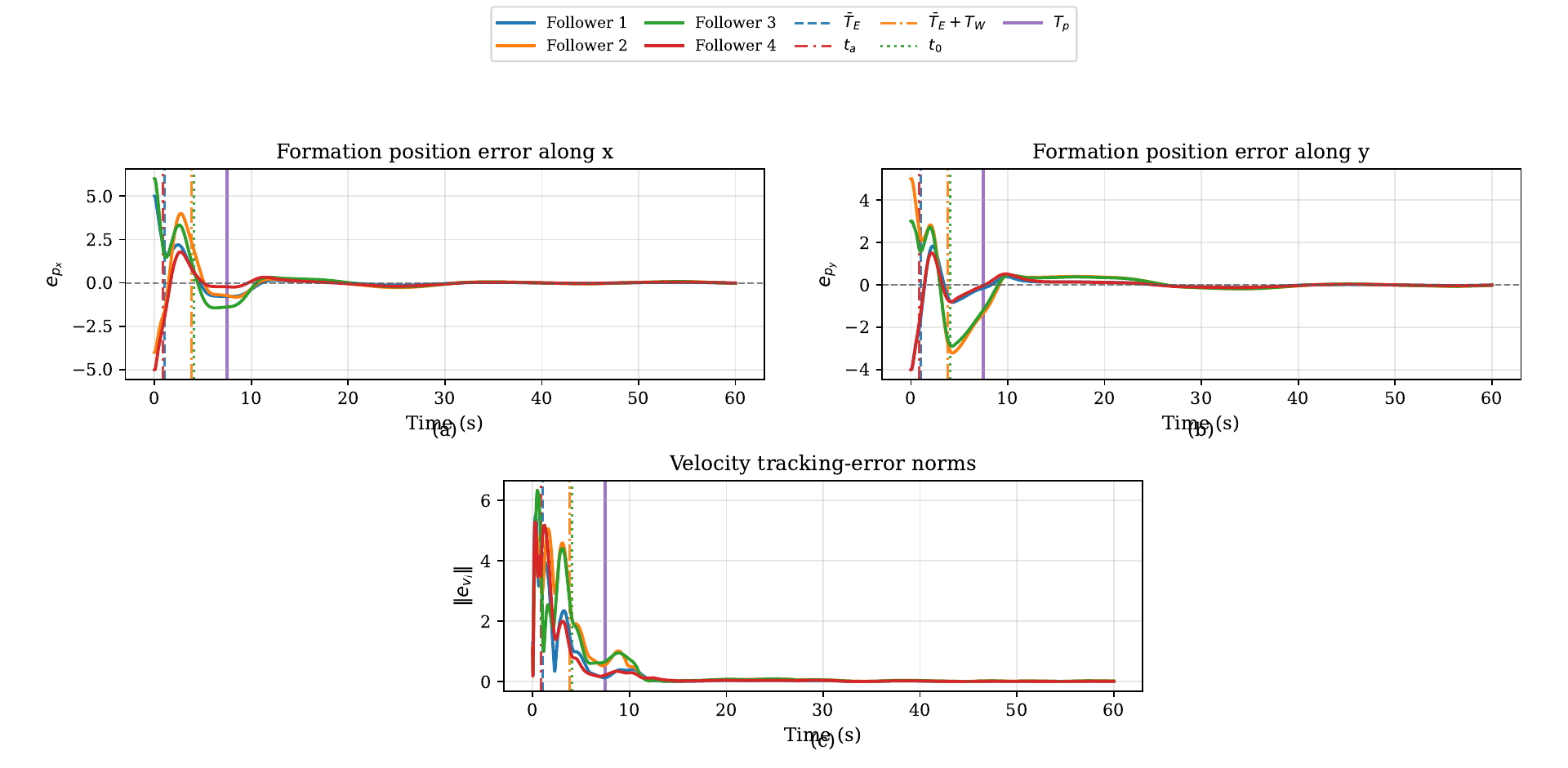}
    \caption{Formation-position errors and velocity-tracking error norms
    of the four followers.}
    \label{fig:formation_errors}
\end{figure}

Figure~\ref{fig:formation_errors} directly verifies the closed-loop
regulation result. The initial position errors reach approximately
\(5\)--\(6\), whereas the largest velocity-error norm is of the same
order. Despite the subsequent FDI and disturbance inputs, all error
channels decrease rapidly and enter a small neighborhood of the origin
before \(T_p\). Small residual excursions remain during the nonideal
transient, which is consistent with the practical predefined-time result:
the theorem guarantees entrance into a residual set determined by critic,
approximation, disturbance, and adversarial mismatches rather than exact
zero convergence under nonvanishing perturbations. As the finite-energy
exogenous signals decay, the observed error neighborhood contracts
accordingly.



\begin{figure}[H]
    \centering
    \includegraphics[
        width=\columnwidth,
        height=0.13\textheight
    ]{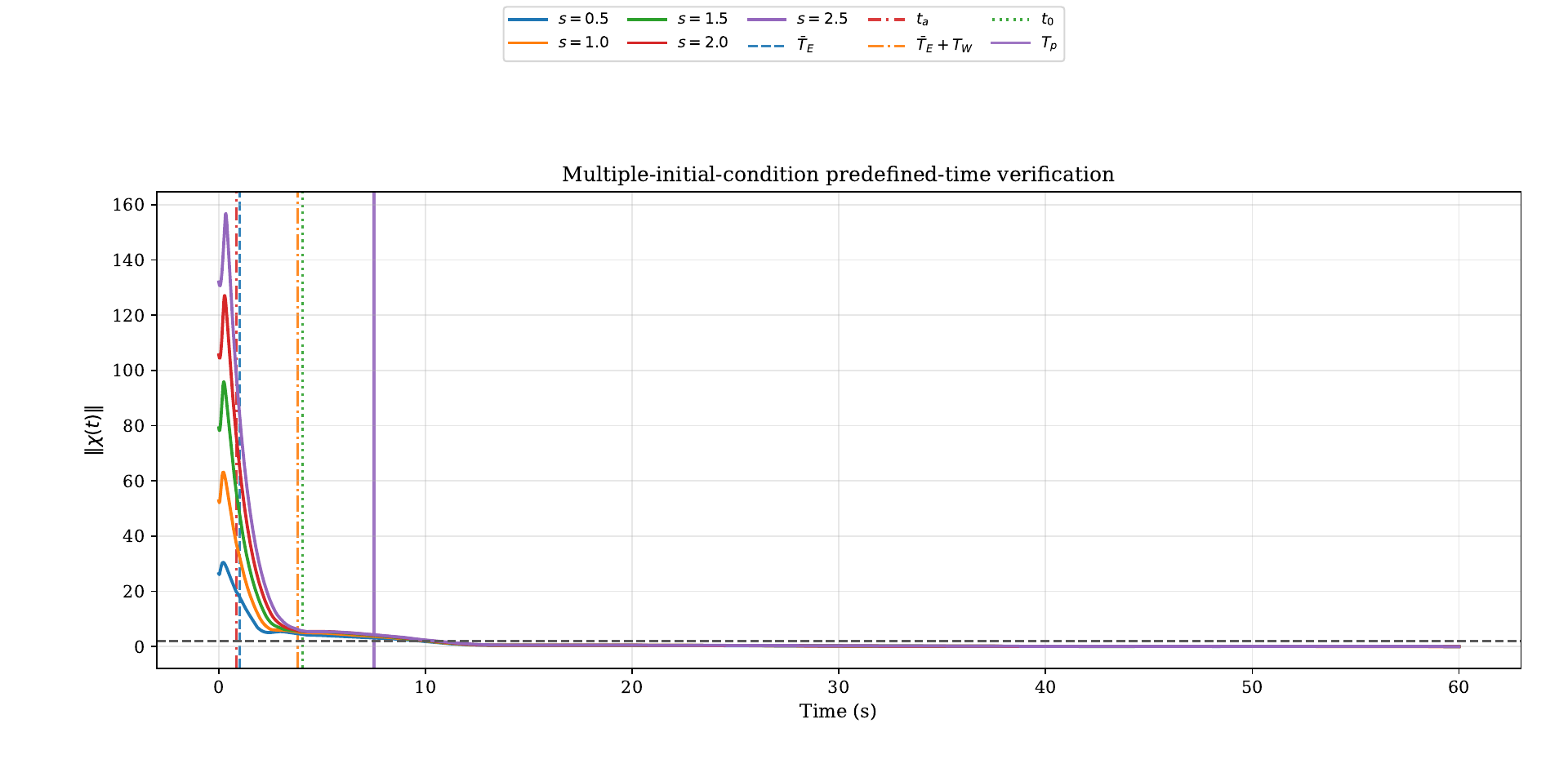}
    \caption{Predefined-time verification under five initial-error
    scaling factors.}
    \label{fig:multiple_initial_conditions}
\end{figure}



The initial-condition independence of the prescribed horizon is examined
in Fig.~\ref{fig:multiple_initial_conditions} using
\(s\in\{0.5,1.0,1.5,2.0,2.5\}\). Although the initial value of
\(\|\chi(t)\|\) increases from approximately \(30\) to more than \(150\),
all trajectories exhibit the same qualitative two-power decay and enter the
small residual region within the single designer-selected deadline \(T_p\).
This highlights the key distinction from conventional finite-time behavior:
the deadline is fixed before operation and is not recomputed from the
initial condition.

Larger initial conditions mainly increase the early control demand and the
high-order dissipation magnitude, rather than proportionally extending the
convergence interval. This is consistent with the role of the superlinear
term in preventing large initial errors from generating arbitrarily long
transients.


\begin{figure}[H]
    \centering
    \includegraphics[width=\columnwidth]
    {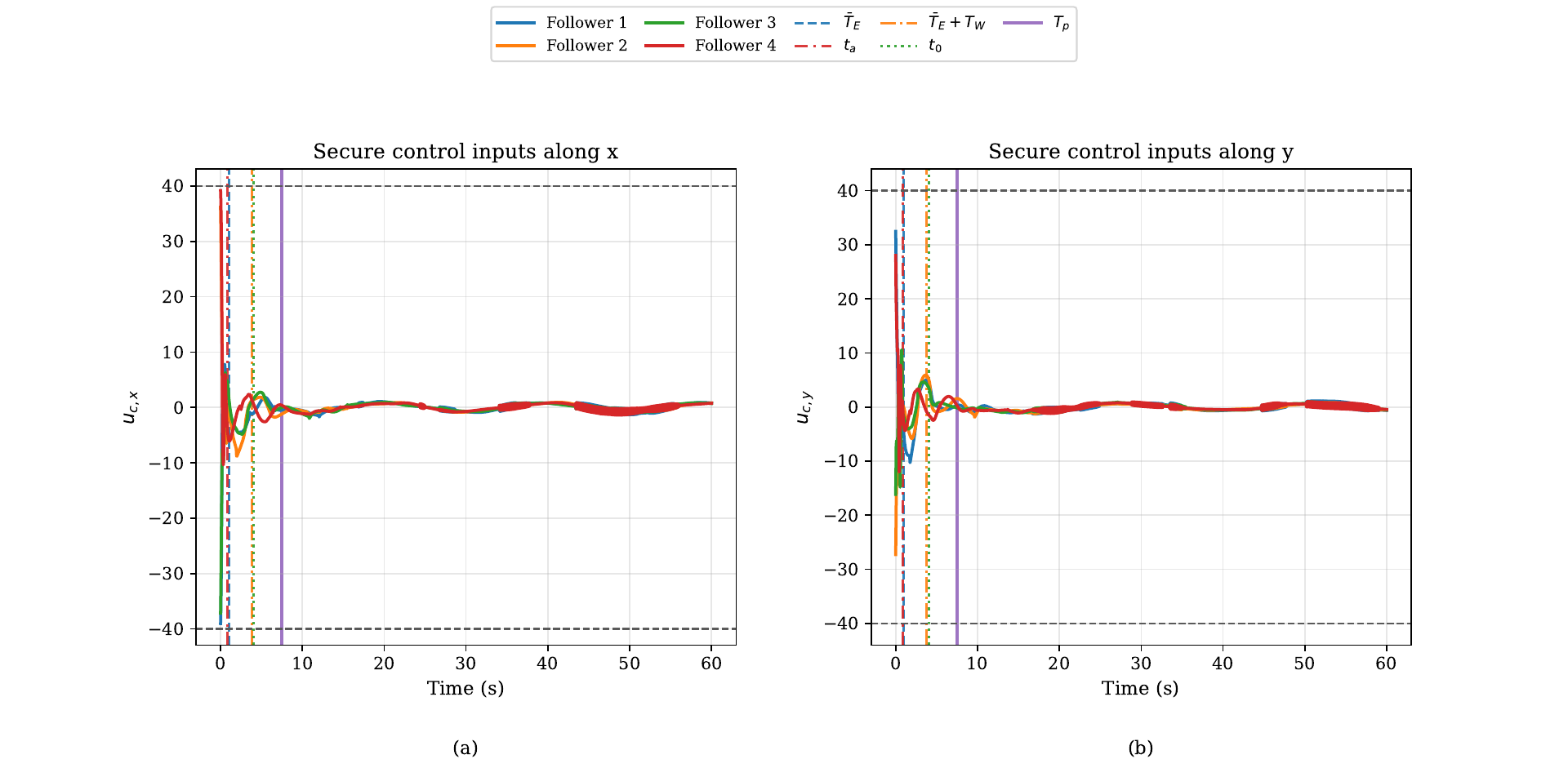}
    \caption{Secure control inputs generated by the
    saturation-compatible policies.}
    \label{fig:secure_control}
\end{figure}

The secure inputs are reported in Fig.~\ref{fig:secure_control}. Several
channels approach the actuator boundary during the severe initial
transient, whereas the subsequent control effort remains much smaller.
Crucially, \(|u_{ci,\ell}(t)|<\bar u_{i\ell},
    \qquad
    \bar u_{i\ell}=40,\)
is preserved throughout the simulation, in agreement with the bounded
hyperbolic-tangent HJI policy. Hence, the predefined-time response is not
obtained by allowing an unbounded control action.

This result also gives a physical interpretation to the inverse-time
design. Decreasing \(T_p\) increases the required state-side gains, but
once the secure command approaches \(\bar u_{i\ell}\), further gain
increases cannot provide proportional actuator authority. Therefore,
although the mathematical synthesis permits the deadline to be assigned
through the derived gain inequalities, a physically meaningful deadline
must remain compatible with actuator capability. Quantifying this minimum
achievable deadline under hard input limits is an important practical
extension of the present theory.


\begin{figure}[H]
    \centering
    \includegraphics[width=\columnwidth]
    {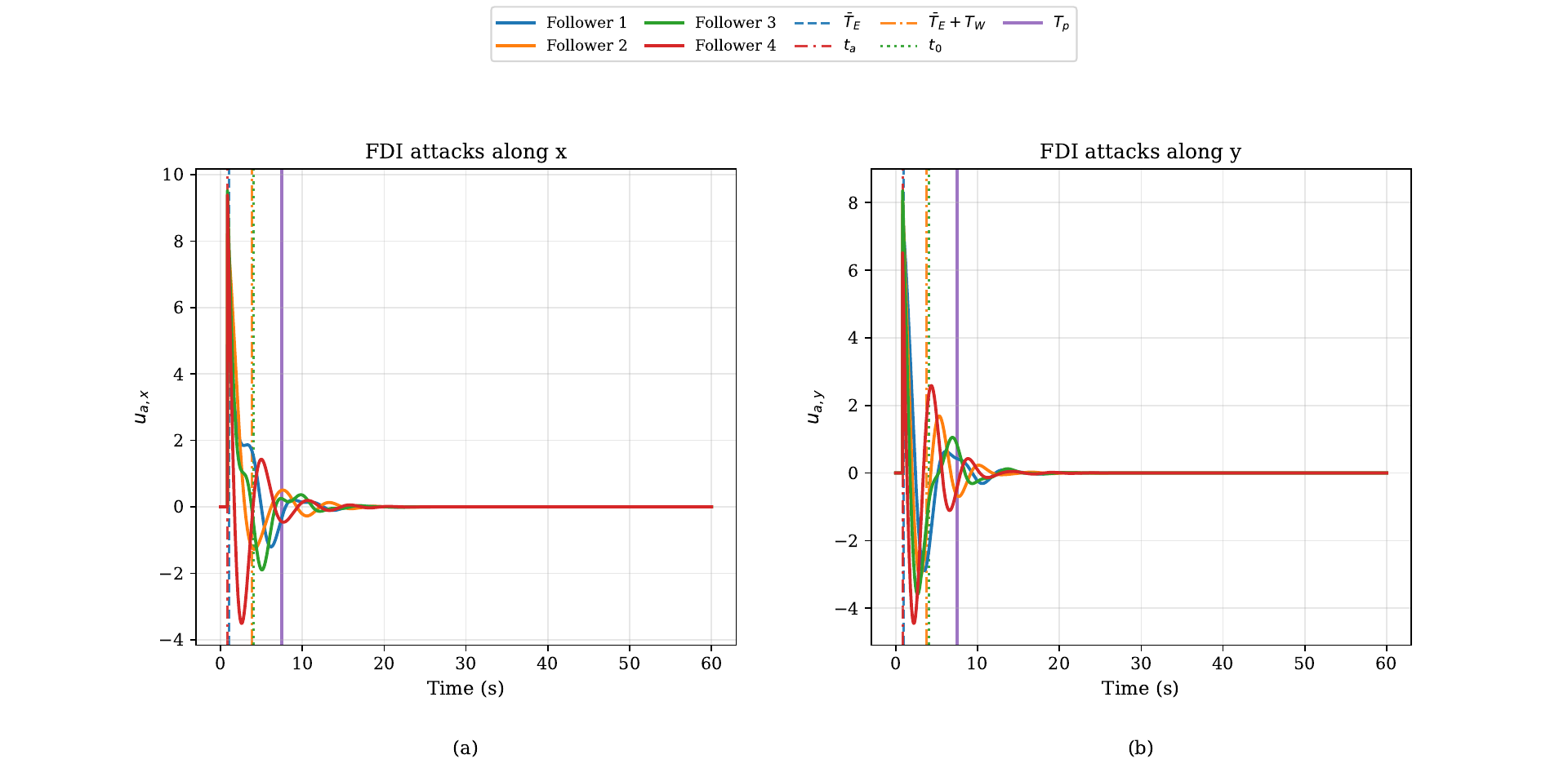}
    \caption{Actuator-side FDI attack signals applied immediately after
    the replay-data acquisition stage.}
    \label{fig:fdi_attacks}
\end{figure}


\begin{figure}[H]
    \centering
    \includegraphics[width=\columnwidth]
    {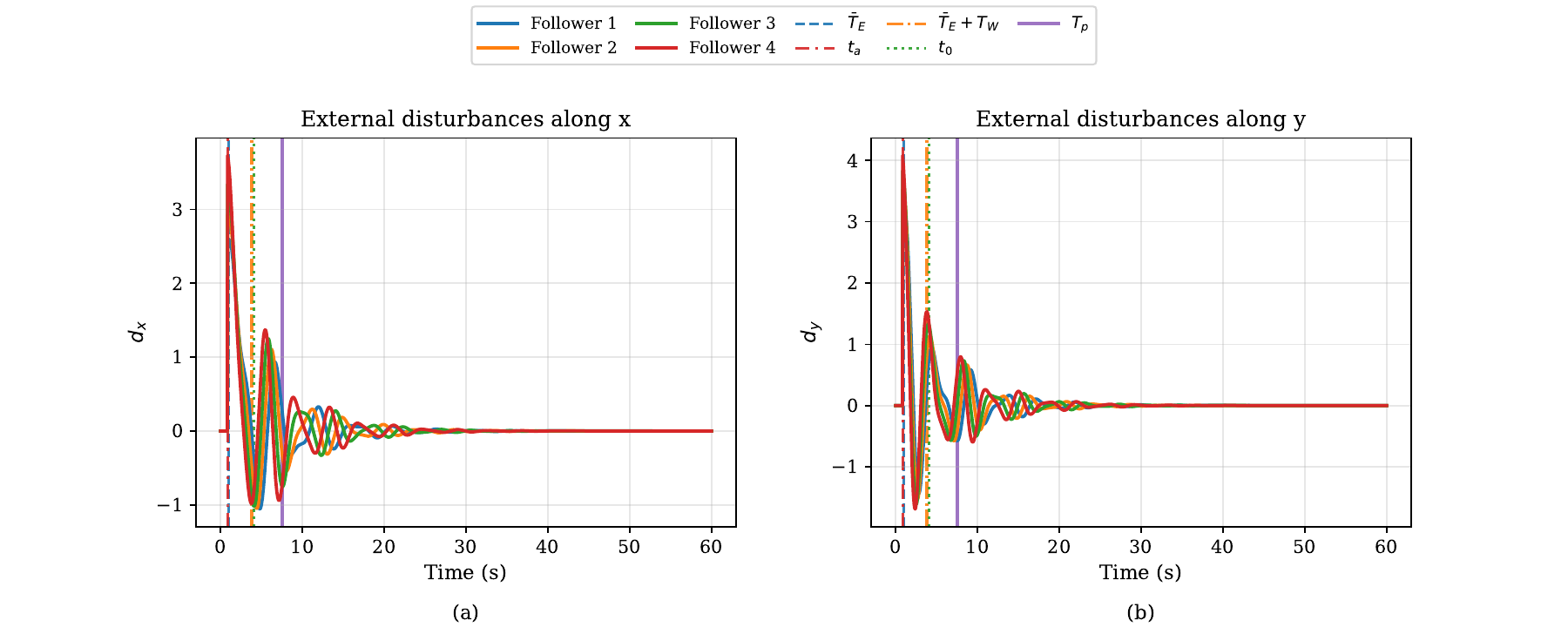}
    \caption{External disturbances acting on the followers, where $d_{i,x}$ and $d_{i,y}$ denote the $x$- and $y$-components of the disturbance signal $\boldsymbol{\omega}_i$, respectively.}
    \label{fig:external_disturbances}
\end{figure}

Figures~\ref{fig:fdi_attacks} and~\ref{fig:external_disturbances} jointly
characterize the adversarial environment used to evaluate the resilient
closed loop. The actuator-side FDI signals and follower-dependent external
disturbances are both bounded, exponentially decaying, and finite-energy,
hence consistent with the assumed
\(L_2([0,\infty))\cap L_\infty([0,\infty))\) signal classes. Both
perturbations are active during the certified transient rather than after
formation has been established. Nevertheless, as verified by
Figs.~\ref{fig:state_tracking} and~\ref{fig:formation_errors}, the state and
formation errors remain bounded and contract as the perturbation energy
decays, consistently with the disturbance- and attack-dependent residual
terms appearing in the Lyapunov analysis.

Overall, the simulations verify the main theoretical mechanisms rather than
only nominal tracking. Finite replay informativity enables sustained
two-power critic adaptation without persistent excitation, while the learned
bounded HJI policies preserve the secure actuator constraints. The resulting
state--critic closed loop remains practically regulated under simultaneous
FDI attacks and disturbances, and widely different initial coordination
errors satisfy the same prescribed deadline, confirming the
initial-condition-independent comparison result.

The results further reveal a tradeoff among deadline, learning aggressiveness,
and actuation. Smaller \(T_W\) and \(T_p\) require larger inverse-designed
critic and state gains, yielding faster convergence at the cost of larger
critic excursions, numerical stiffness, and actuator utilization. Hence,
deadline selection should account jointly for actuator authority, replay
conditioning, sampling rate, and admissible adversarial energy. This motivates
explicit lower bounds on physically feasible predefined times and tighter
residual-set estimates, together with future extensions to sampled-data
communication, finite-rate actuators, measurement noise, and experimental
networked robotic platforms.

\endgroup

\section{Conclusion}
\label{sec:conclusion}


This paper proposed a predefined-time resilient IRL framework for unknown
nonlinear leader--follower multi-agent systems under actuator constraints,
disturbances, and FDI attacks. A local zero-sum HJI game unified secure and
adversarial channels, while a saturation-aware nonquadratic utility ensured
bounded policies. Critic-only learning was achieved through an integral
Bellman--Isaacs residual without drift knowledge, and replay data removed the
need for persistent excitation. A two-power cost, predefined-time critic
learning, and pointwise--integral Lyapunov analysis yielded explicit gains
ensuring practical critic and formation convergence within a designer-assigned,
initial-condition-independent deadline. Simulations verified the results under
different initial conditions, actuator constraints, disturbances, and FDI
attacks. Future work will consider communication constraints, switching
topologies, and experiments.

\section*{Acknowledgment}
The authors would like to thank Ho Chi Minh City University of Technology (HCMUT) and Vietnam National University Ho Chi Minh City (VNU-HCM) for supporting this research.

\raggedbottom

\makeatletter
\def\@IEEEBIOskipN{0.5\baselineskip}
\makeatother

\vspace{0.3em}
\section*{Author Biographies}
\vspace{-0.6em}


\begin{IEEEbiography}
[{\includegraphics[
    width=1in,
    height=1.25in,
    clip,
    keepaspectratio
]{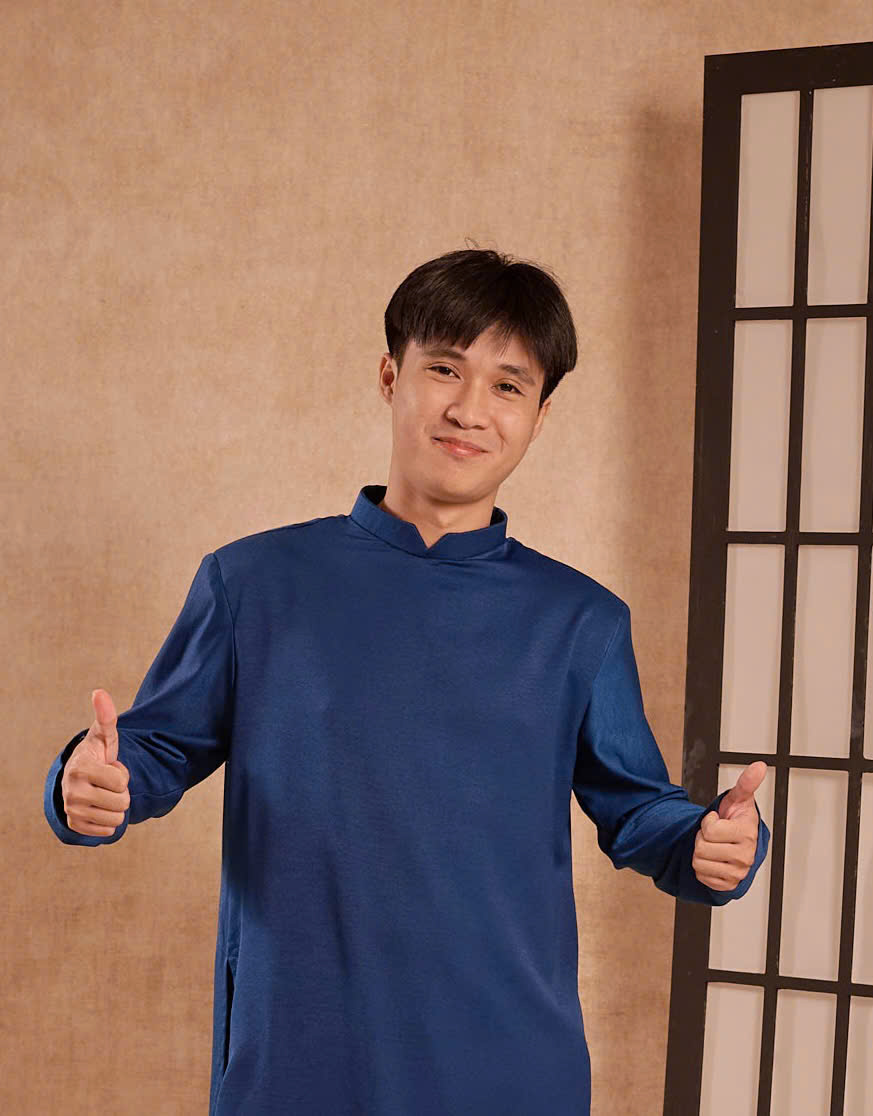}}]
{Tien Dat Vu}
Tien Dat Vu is currently a senior undergraduate student in the Vietnamese–French Program in Mechatronics Engineering at Ho Chi Minh City University of Technology (HCMUT), Vietnam National University Ho Chi Minh City (VNU-HCM), Ho Chi Minh City, Vietnam. His research interests include control theory, learning-based and nonlinear control, adaptive and robust control (Convex optimization in control), reinforcement learning, data-driven control, secure and resilient control, optimal control of uncertain nonlinear systems, multi-agent systems, and artificial intelligence for dynamical systems. His broader research interests also include Riemannian geometry, contraction theory, formal methods, and their applications to nonlinear dynamical systems and robotics.
\end{IEEEbiography}


\begin{IEEEbiography}
[{\includegraphics[
    width=1in,
    height=1.25in,
    clip,
    keepaspectratio
]{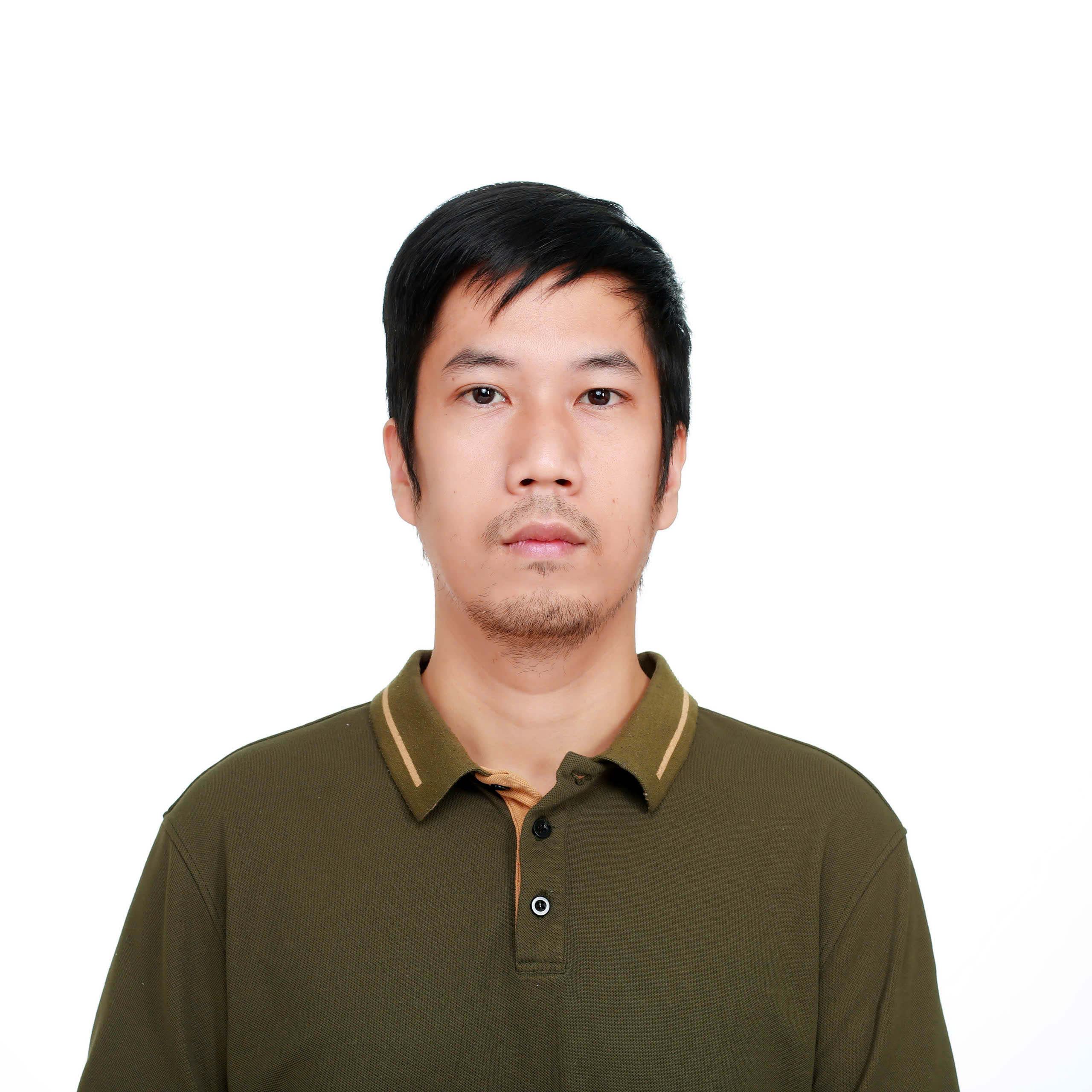}}]
{Nhat Minh Doan}
received the B.S. degree in Mechanical Engineering from
Bucknell University, Lewisburg, PA, USA, in 2015, and the
Ph.D. degree from Keio University, Tokyo, Japan, in 2021.
He is currently a Lecturer with the Faculty of Mechanical
Engineering, Ho Chi Minh City University of Technology
(HCMUT), Vietnam National University Ho Chi Minh City
(VNU-HCM), Ho Chi Minh City, Vietnam. His research interests
include unmanned aerial vehicle design and control, wind
turbine systems, multi-agent systems, networked and distributed
control, and the modeling, analysis, and control of complex
mechanical and networked dynamical systems.
\end{IEEEbiography}

\vspace{0.5em}

\end{document}